\documentclass[lettersize,journal]{IEEEtran}
\usepackage{amsmath,amsfonts}
\usepackage{algorithmic}
\usepackage{algorithm}
\usepackage{array}
\usepackage[caption=false,font=footnotesize,labelfont=rm,textfont=rm]{subfig}
\usepackage{textcomp}
\usepackage{stfloats}
\usepackage{url}
\usepackage{verbatim}
\usepackage{graphicx}
\usepackage{cite}
\usepackage{color}
\usepackage[colorlinks=true, allcolors=blue]{hyperref}

\let\oldeqref\eqref 
\renewcommand{\eqref}[1]{\textcolor{blue}{\oldeqref{#1}}} 
\usepackage{amssymb}
\usepackage{makecell}
\usepackage{amsthm}

\newtheorem{proposition}{Proposition}

\newtheorem{lemma}{Lemma}
\makeatletter
\renewenvironment{proof}[1][\proofname]{\par
	\pushQED{\qed}%
	\normalfont\topsep0\p@\relax
	\trivlist
	\item[\hskip\labelsep
	\itshape 
	#1\@addpunct{:}]\ignorespaces  
}{
	\popQED\endtrivlist\@endpefalse
}

\renewcommand{\proofname}{Proof}

\makeatother
\usepackage{multirow}
\usepackage{diagbox}
\usepackage{threeparttable}
\begin{document}
	
	\title{On Performance of LoRa Fluid Antenna Systems}
	
	\author{Gaoze Mu, Yanzhao Hou,~\IEEEmembership{Member,~IEEE,} Kai-Kit Wong,~\IEEEmembership{Fellow,~IEEE,} Mingjie Chen,\\ Qimei Cui, \IEEEmembership{Senior Member,~IEEE,} Xiaofeng Tao, \IEEEmembership{Senior Member,~IEEE,} and Ping Zhang, \IEEEmembership{Fellow,~IEEE}
		
		\thanks{\textit{Corresponding author: Yanzhao Hou.}}
		
		\thanks{Gaoze Mu and Mingjie Chen are with the National Engineering Research Center for Mobile Network Technologies, Beijing University of Posts and Telecommunications, Beijing 100876, China (e-mail: \{mugz; chenmingjie\}@bupt.edu.cn).}
		
		\thanks{Yanzhao Hou, Qimei Cui, and Xiaofeng Tao are with the National Engineering Research Center for Mobile Network Technologies, Beijing University of Posts and Telecommunications, Beijing 100876, China, and also with the Department of Broadband Communication, Peng Cheng Laboratory, Shenzhen 518055, China (e-mail: \{houyanzhao; cuiqimei; taoxf\}@bupt.edu.cn).}
		
		\thanks{Ping Zhang is with the State Key Laboratory of Networking and Switching Technology, Beijing University of Posts and Telecommunications, Beijing 100876, China (e-mail: pzhang@bupt.edu.cn).}
		
		\thanks{Kai-Kit Wong is with the Department of Electronic and Electrical Engineering, University College London, WC1E 6BT London, UK, and also with the Department of Electronic Engineering, Kyung Hee University, Yongin-si, Gyeonggi-do 17104, Korea (e-mail: kai-kit.wong@ucl.ac.uk).}
	}
	
	\markboth{}%
	{Shell \MakeLowercase{\textit{et al.}}: A Sample Article Using IEEEtran.cls for IEEE Journals}
	
	\IEEEpubid{}
	
	\maketitle
	
	\begin{abstract}
		\textcolor{black}{This paper advocates a fluid antenna system (FAS)-assisted long-range communication (LoRa-FAS) for Internet-of-Things (IoT) applications. In the proposed system, FAS provides spatial diversity gains for LoRa, eliminating the necessity for integrating multiple-input multiple-output (MIMO) technologies into the system. It consists of a traditional LoRa transmitter with a fixed-position antenna and a LoRa receiver employing the FAS (Rx-FAS). The pilot sequence overhead and placement for FAS are also considered. Specifically, we consider embedding pilot sequences within symbols to reduce the impact of pilot overhead on system throughput and the physical layer (PHY) frame structure, leveraging the fact that the pilot sequences do not convey source information and correlation detection at the LoRa receiver needs not be performed across the entire symbol. The achievable performance of LoRa-FAS is thoroughly analyzed under both coherent and non-coherent detection schemes. We obtain new closed-form approximations for the probability density function (PDF) and cumulative distribution function (CDF) of the FAS channel under the block-correlation model. Furthermore, the approximate SER, equivalently the bit error rate (BER), of the proposed LoRa-FAS is also derived in closed form.} Simulation results indicate that substantial SER gains can be achieved by FAS within the LoRa framework, even with a limited size of FAS. In addition, our analytical results align well with Clarke's exact spatial correlation model. Finally, when utilizing the block-correlation model, we suggest that the correlation factor should be selected as the proportion of the eigenvalues of the exact correlation matrix greater than $1$ for higher accuracy. 
	\end{abstract}
	
	
	\section{Introduction}
	\subsection{Background}
	\IEEEPARstart{E}{merging} applications in the Internet-of-Things (IoT) have heightened the technical demands on current Low-Power Wide-Area Networks (LPWANs) \cite{bibitem1}. Amongst existing physical (PHY) layer solutions for LPWANs, such as narrowband (NB)-IoT \cite{bibitem2}, random phase multiple access (RPMA) and long-range (LoRa) \cite{bibitem3}, Semtech's LoRa \cite{bibitem4} is particularly attractive due to its supreme capability to support long-range, low-cost, energy-efficient, and anti-jamming communications, enabled by its specific chirp spread spectrum (CSS) modulation technique. In LoRa modulation, cyclic shifts of a fundamental chirp signal, characterized by spreading factors (SFs) and linearly increasing frequency, are employed to create orthogonal signals and distribute the signal energy across a wider bandwidth (BW) \cite{bibitem5}. This approach enhances system robustness against noise and interference \cite{bibitem6}. \textcolor{black}{By employing different SFs, CSS modulation enables multiple IoT devices to communicate simultaneously with minimal mutual interference. Additionally, CSS can support a wide range of IoT devices with diverse quality of service (QoS) requirements and capabilities through the adjustment of system parameters, including SF, coding rate (CR), and BW \cite{bibitem7}.}
	\textcolor{black}{The CSS scheme provides notable advantages to LoRa systems, yet it remains subject to several limitations. As in other communication systems, an inherent trade-off exists between frequency resource occupancy, data transmission rate, and reliability. In particular, increasing the SF can significantly enhance the communication range or reduce the error probability, albeit with a decreased transmission rate. Moreover, a larger bandwidth results in shorter symbol duration, enabling more data to be transmitted within a given time interval. More importantly, the analytical results in \cite{bibitem8} highlight severe challenge encountered by LoRa under fading conditions. Specifically, for a target bit error rate (BER) of $10^{-4}$, LoRa experiences over 30 dB performance degradation when operating in a Rayleigh fading channel compared to an additive white Gaussian noise (AWGN) channel. The measurement results in \cite{bibitem9} indicate that the coverage range of LoRa is approximately 8 km in urban areas and up to 45 km in rural areas. To resist the severe fading, additional techniques are needed to enhance CSS modulation and LoRa system, without compromising other performance of LoRa.}
	\IEEEpubidadjcol
	\subsection{Related Works on LoRa-MIMO}
	Multiple-input-multiple-output (MIMO) is widely regarded as an effective means to achieve diversity or multiplexing gains through spatial exploitation. \textcolor{black}{However, their applicability to LoRa systems is quite constrained due to the characteristics of CSS signals, hardware requirements, antenna spacing limitations at lower carrier frequency (e.g., 433.05-434.79 MHz and 779-787 MHz \cite{bibitem10}), and fundamentally the design philosophy. Specifically, LoRa is characterized by its lower power consumption, reduced data rates and lower hardware complexity, whereas MIMO techniques are typically associated with high data rates, increased hardware complexity, and greater power consumption.} Few studies have investigated the integration of MIMO and LoRa systems \cite{bibitem11}, \cite{bibitem12}, \cite{bibitem13}, \cite{bibitem14}, \cite{bibitem15}, \cite{bibitem16}, \cite{bibitem17}. \textcolor{black}{Despite the aforementioned challenges, the spatial gain provided by MIMO to LoRa is considerable.} These studies applied MIMO to improve link reliability over fading channels or to support higher-rate transmissions. Specifically, \cite{bibitem11}, \cite{bibitem12} investigated incorporating space-time block coding (STBC) MIMO schemes into LoRa systems, which is able to achieve the full diversity order and enhanced BER performance. The MIMO scheme with receive combining has also been demonstrated to effectively enhance the reliability of LoRa \cite{bibitem13}. Others have investigated techniques to achieve high-rate LoRa by transmitting multiple modulated chirp signals across antennas with varying SFs \cite{bibitem14}, \cite{bibitem15}, \cite{bibitem16}. \textcolor{black}{Moreover, the most recent work points out that $4\times N_\text{r}$ MIMO-CSS with permutation matrix modulation (PMM) can realize $440\%$ spectral efficiency (SE) compared with the basic CSS modulation \cite{bibitem17}. Notably, high-rate LoRa-MIMO relies to some extent on the inherent orthogonality of the LoRa symbols or the use of high-complexity coherent detection methods. Returning to the practical scenario, the aforementioned issue continues to raise doubts about the feasibility of implementing MIMO for LoRa, and an alternative scheme that can provide additional spatial opportunity for LoRa without MIMO is preferred.}
	\subsection{Related Works on FAS}
	Following the above discussion, recent advances in reconfigurable antenna technologies do give us hope. In particular, the concept of fluid antenna system (FAS) that advocates position and shape flexibility in antennas, has emerged as a new degree of freedom in the physical layer of wireless communications \cite{bibitem18}, \cite{bibitem19}. In \cite{bibitem20}, \cite{bibitem21}, Wong {\em et al.}~revealed that even with a small space at user equipment, tremendous performance gains can be obtained by exploiting position flexibility on one RF chain. It is worth mentioning that recent terminology such as movable antenna systems also falls under the category of FAS \cite{bibitem22}. FAS may be realized by using liquid materials \cite{bibitem23}, metamaterials \cite{bibitem24}, reconfigurable RF pixels, \cite{bibitem25} and etc. Experimental results on FAS were reported recently in \cite{bibitem26} and \cite{bibitem27}. For a comprehensive review covering various aspects of FAS, readers are referred to \cite{bibitem28}.
	
	
	FAS can also be really useful to multiuser communications, especially when channel state information (CSI) is not available on the transmitter side, so MIMO precoding is not an option. This has given rise to the fluid antenna multiple access (FAMA) technology \cite{bibitem29}, which utilizes the fading phenomenon in the spatial domain to avoid interference. Also, FAMA has been considered in conjunction with opportunistic scheduling that can significantly enhance its interference immunity \cite{bibitem30}, \cite{bibitem31}. FAS was effective in improving energy efficiency for non-orthogonal multiple access (NOMA) systems \cite{bibitem32}. Moreover, FAS has expanded its use to a variety of emerging applications, such as simultaneous wireless information and power transfer (SWIPT) \cite{bibitem33}, integrated sensing and communication (ISAC) \cite{bibitem34} and physical layer security \cite{bibitem35}, among others. On the other hand, reconfigurable intelligent surface (RIS) \cite{bibitem36} and artificial intelligence (AI) or machine learning \cite{bibitem37} can synergize with FAS for significant effects. \textcolor{black}{Additionally, CSI is pivotal for FAS performance and thus channel estimation has been an important topic of research for FAS. Recent works \cite{bibitem38}, \cite{bibitem39} have provided useful channel estimation frameworks that can be integrated into FAS. Given the existing results for FAS, it may be a suitable approach for communication systems that are challenged in achieving spatial gains through conventional MIMO techniques.}
	
	\subsection{Motivation and Key Contributions}
	\textcolor{black}{Motivated by recent advancements in LoRa technology and the promising potential of FAS, we recognize FAS as a suitable approach to achieve spatial diversity gains for LoRa, which would otherwise require integration with MIMO architectures. Unlike traditional MIMO systems, the simplest implementation of FAS requires only a single RF chain. When a given antenna port is selected and activated, FAS operates as a conventional single-position antenna system, enabling seamless integration with LoRa. On the other hand, the low carrier frequency used by LoRa necessitates a large antenna spacing for conventional MIMO systems, which is often impractical. In contrast, FAS can exploit spatial diversity even at sub-half-wavelength scales, making it more feasible for LoRa deployments. Note that recent technologies, such as holographic MIMO (H-MIMO) \cite{bibitem40} and compact ultra massive arrays (CUMA) \cite{bibitem41}, can also realize spatial potential at sub-half-wavelength intervals. However, they remain prohibitively complex for LoRa applications. It is important to note that while the available physical space may limit the achievable spatial gain, correlation detection at the LoRa receiver can aggregate gains from multiple sampling points of LoRa symbols. This capability significantly enhances performance even within small spatial dimensions. Furthermore, FAS eliminates the need for complex RF hardware, advanced precoding algorithms, and elaborate protocol stacks, making it well-aligned with the design principles of LoRa. From an analytical perspective, accurate channel models are essential for the performance evaluation of FAS, while the precision and tractability of these models remain an open problem. The block-correlation model \cite{bibitem42} and other eigenvalue-based models \cite{bibitem43} offer higher precision but introduce complexities in performance analysis. However, it keeps the tractability of the simpler model used in \cite{bibitem21}, \cite{bibitem22} with decreased precision. For a comprehensive review covering various channel models of FAS, readers can refer to \cite{bibitem28}. To prove the potential of FAS in the LoRa system, analytical expressions for the FAS channel that offer an appropriate trade-off between accuracy and tractability are also anticipated.}
	
	In this paper, we investigate the integration of FAS with LoRa and analyze its statistical performance. In summary, we have made the following contributions.
	\begin{itemize}
		\item \textcolor{black}{We propose a FAS-assisted LoRa (LoRa-FAS) designed to mitigate severe performance degradation caused by channel fading and to enhance link reliability. Specifically, the system adopts a standard Rx-SISO-FAS architecture, where a traditional fixed-position antenna is utilized at the transmitter, and the receiver is equipped with a FAS. Additionally, the standard LoRa modulator is employed, and both non-coherent and coherent LoRa detectors are considered.}
		\item \textcolor{black}{An embedded symbol pilot placement strategy is proposed to reduce the impact of pilot overhead on system throughput. The proposed scheme leverages that pilots do not convey information, and the symbol detection does not need to be executed across the entire LoRa symbol. The proposed scheme preserves the frame structure of the existing LoRa PHY layer.}
		\item \textcolor{black}{We obtain the approximate probability density function (PDF) and cumulative distribution function (CDF) of the correlated FAS channel under the block-correlation model and present them in an analytical and tractable form. Numerical results demonstrate that the proposed expressions closely match the Monte-Carlo simulations based on the exact Clarke’s correlation model.}
		\item \textcolor{black}{We derive approximate expressions for SER (and hence BER) of the proposed LoRa-FAS under non-coherent and coherent receivers in closed form.}
		\item When using the block-correlation model in \cite{bibitem42} to derive the analytical results, our numerical results reveal that a judicious choice for the correlation parameter is given by the proportion of the exact correlation matrix's eigenvalues greater than $1$ for great modeling accuracy.
		\item Numerical results \textcolor{black}{demonstrate} that FAS can \textcolor{black}{substantially enhance} the SER performance by several orders of magnitude for LoRa, even if the size of FAS is small.
	\end{itemize}
	
	\emph{Organization:} The remainder of this paper is organized as follows. Section \ref{sec:model} introduces the system model of LoRa-FAS, including the channel correlation model. Subsequently, Section \ref{sec:analysis} presents the main analytical results. The numerical results are provided in Section IV, while the paper is concluded in Section V.
	\begin{table}[t]\color{black}
		\centering
		\caption{The meanings of key variables}
		\begin{tabular}{m{0.9cm}<{\centering}||l}
			\hline
			{Notation } & \makecell[c]{Meaning} \\ \hline
			$\Gamma$ & LoRa SNR \\ \hline
			$N_0$ & Single-side noise power spectral density \\ \hline
			$m$ & Transmit symbol with information $m$\\ \hline
			$n$ & $n$-th discrete sampling point\\ \hline
			$k$ & $k$-th frequency bin in LoRa demodulation\\ \hline
			\multirow{2}{*}{$M$} & Total number of different transmit symbols, \\ 
			&discrete sampling points and frequency bins\\ \hline
			\multirow{2}{*}{$u;U$} & $u$-th symbol in the observation cycle; \\
			&  Total number of symbols in the observation cycle\\ \hline
			$W$ & Length of the fluid antenna in terms of wavelength \\ \hline
			$l;L$ & $l$-th FAS port; Total number of ports\\ \hline
			$b;B$ & $b$-th block in block diagonal matrix; Total number of blocks\\ \hline
			\multirow{2}{*}{$l_b;L_b$} & $l_b$-th element in $b$-th block; \\ 
			& Total number of elements in $b$-th block\\
			 \hline
			$X^{(\mathrm{ncoh})}_k$ & Decision values for the non-coherent and coherent receivers\\
			$X^{(\mathrm{coh})}_k$ & at the $k$-th frequency bin\\ \hline
		\end{tabular}
		\label{table_key_variables}
	\end{table}

	\emph{Notations:} The following notational conventions are adopted throughout our discussions. Lowercase and uppercase italic letters stand for variables and parameters, respectively. Besides, ${\rm{Corr}\{\cdot,\cdot\}}$ denotes the correlation coefficient. $\rm{Pr}(\cdot)$ represents the probability of an event. $\rm{E}\{\cdot\}$ and $\rm{Var}\{\cdot\}$ stand for the mathematical expectation and variance operator, respectively. $|\cdot|$ denotes the absolute and $\varphi(\cdot)$ denotes the phase. $\|\cdot\|_2$ is the $\ell_2$ norm of a vector. \textcolor{black}{$\cdot^*$ denotes the conjugate of a complex variable. Finally, $\hat{\cdot}$ is the maximum magnitude of a set of variables.} \textcolor{black}{Also, the following special functions are adopted throughout our discussions. ${H}(\cdot)$ is the Heaviside unit step function \cite[1.16.13]{bibitem45}. ${Q}(\cdot)$ and ${Q}_1(\cdot,\cdot)$ are the Gaussian Q-function \cite[(A.1)]{bibitem46} and first-order Marcum Q-function \cite[(A.3)]{bibitem46}, respectively. $\mathcal{W}(\cdot)$ is the Lambert W-function \cite[4.13.1]{bibitem45}. $\mathrm{erfc}(\cdot)$ is the complementary error function \cite[7.2.2]{bibitem45}. $I_0(\cdot)$ is the $0$-th order modified Bessel function \cite[10.25.2]{bibitem45}. $\mathrm{Ra}(\sigma)$ denote the Rayleigh distribution with the scale parameter of $\sigma$. $\mathrm{Ri}(\nu,\sigma)$ is the Rician distribution with shape parameter of $\frac{\nu^2}{2\sigma^2}$. $\mathcal{CN}(\mu,\sigma^2)$ and $\mathcal{N}(\mu,\sigma^2)$ denote the complex Gaussian and Gaussian distribution with mean $\mu$ and variance $\sigma^2$, respectively. $\mathcal{G}(\mu,\beta)$ is the Gumbel distribution with location parameter $\mu$ and scale parameter $\beta$. Table I summarizes the definitions of key variables to help readers understand the system model and mathematical content.} 
	
	\section{LoRa-FAS System Model}\label{sec:model}
	\textcolor{black}{We consider a LoRa-FAS, where a LoRa modulator \cite{bibitem4} generates the modulated chirp signals based on the symbol information from the source. Subsequently, a FAS pilot inserter embeds the pilot sequences into the modulated signals. The LoRa-FAS receiver comprises an $L$-port FAS and a LoRa demodulator. The FAS switches to the optimal port that maximizes the channel magnitude based on the CSI estimated from the pilot sequences. Thereafter, the LoRa demodulator recovers the transmitted symbols from the received signals by employing either the non-coherent or coherent detection method. In the following subsections, we detail the transmitter and receiver design of the proposed LoRa-FAS, as well as the channel model.}
	\subsection{Transmitter}
	\begin{figure}[!t]\color{black}
		\includegraphics[width=1\linewidth]{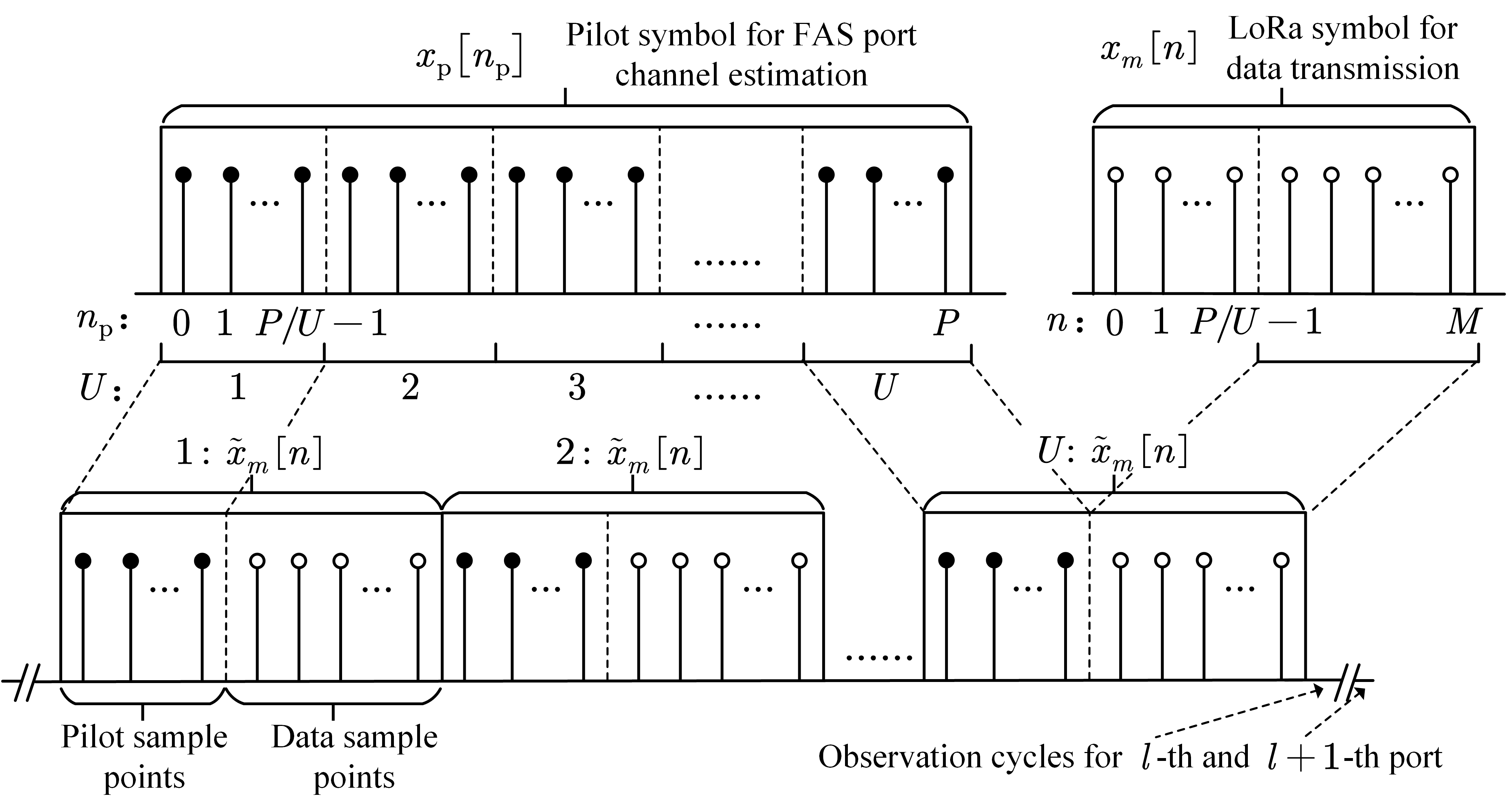}%
		\caption{Schematic diagram of the proposed symbol structure.}
		\label{fig:1}
	\end{figure}
	LoRa is a spread-spectrum $M$-ary modulation scheme that uses a given BW ${\Delta}$ and $M$ possible waveforms, \textcolor{black}{where $M=2^{SF}$ and $SF\in\{7,8,\dots,12\}$.} Typically, \textcolor{black}{the} basic discrete-form chirp signal can be denoted as  \cite{bibitem47}
	\begin{equation}\label{x0}
		\begin{split}
			\color{black}{x_0[nT]}=&\sqrt{\frac{1}{\color{black}{M}}}e^{j2\pi{n T\Delta}\left(\frac{nT\Delta }{2\color{black}{M}}-\frac{1}{2}-H\left({nT}-T_{\rm{s}}\right)\right)}\\
			=&\sqrt{\frac{1}{\color{black}{M}}}e^{j2\pi\left(\frac{n^2}{2\color{black}{M}}-\frac{n}{2}\right)}\color{black}{\equiv x_0[n]},\\
		\end{split}
	\end{equation}
	\textcolor{black}{for $n=0,1,...,M-1$.} \textcolor{black}{Besides}, $T=\frac{1}{\Delta}$ denotes the sample interval and the symbol duration is $\color{black}{T_{\rm{s}}=M\cdot T}$. The modulated basic chirp with the information $m\in{\color{black}{\mathcal{M}}}=\{0,1,\dots,M-1\}$ is given by 
	\begin{equation}\label{xm}
		x_m[n]=\sqrt{\frac{1}{\color{black}{M}}}e^{j2\pi\left(\frac{n^2}{2\color{black}{M}}-\frac{n}{2}+\frac{mn}{\color{black}{M}}\right)},
	\end{equation}
	\textcolor{black}{which performs frequency-shift modulation on the basic chirp symbol.} \textcolor{black}{From the perspective of FAS,} the CSI for all ports should be estimated to perform the switching operation. The pilot \textcolor{black}{symbol} for each port \textcolor{black}{is set as} the unmodulated chirp symbol with \textcolor{black}{a SF of} \textcolor{black}{ $\color{black}{SF_\text{p}\in\{7,8,\cdots,12\}}$.} Let $\color{black}{P=2^{SF_\text{p}}}$, the pilot symbol can be expressed as 
	\begin{equation}\label{pilot}
		x_{\rm{p}}[n]=\sqrt{\frac{1}{\color{black}{M}}}e^{j2\pi\left(\frac{\color{black}{n_{\text{p}}}^2}{2\color{black}{P}}-\frac{\color{black}{n_{\text{p}}}}{2}\right)},
	\end{equation}
	\textcolor{black}{for $n_{\rm{p}}=0,1,...,P-1$.} In order to retain the existing frame structure of the LoRa PHY layer, we discretize one pilot symbol into $U$ LoRa symbols, \textcolor{black}{where $U>\frac{P}{M}$}, substituting the first \textcolor{black}{$\frac{P}{M\cdot U}$} portion of the LoRa symbol.\footnote{Although this operation necessitates a symbol-level switching rate for FAS ports, it is considered practical within LoRa systems, given their lower sampling frequency and extended symbol duration ranging from 125 to 500 kHz and 36 to 682 ms, respectively, compared to the 15.36 to 122.88 MHz and 4.17 to 66.67 $\mu$s in modern 5G networks.} \footnote{\textcolor{black}{The flat Rayleigh fading is the most practical and prevalent channel model for LoRa Systems, particularly in scenarios involving long-distance transmission, narrow-bandwidth (e.g., 125, 250, and 500 kHz), and low mobility. Under this assumption, other research works employ CSI for coherent detection \cite{bibitem48}, \cite{bibitem49} or STBC \cite{bibitem11}, \cite{bibitem12} typically assume that CSI can be obtained through the LoRa preamble at the beginning of each PHY frame. In LoRa-FAS, considering the channel coherence time and the convenience of implementation, a portion of the pilot sequences is placed at the beginning of each LoRa symbol. This modified pilot placement can be effectively embedded within the symbol without distinguishing between frame components.}}  \textcolor{black}{For simplicity, we assume $\frac{U\cdot M}{P}$ is an integer. The proposed symbol structure is illustrated in Fig. \ref{fig:1}.} Then an $L$-port FAS requires $L\cdot U$ symbols to scan all ports.\footnote{In this context, we consider ideal switching for FAS, as any overhead can be incorporated after the pilot sequences without affecting the subsequent analytical process. \textcolor{black}{Note that fluid-material-based FAS are unsuitable due to the high acceleration during movement \cite{bibitem41}, whereas micro-electromechanical systems (MEMS)-based pixel antennas offer a more practical solution. MEMS devices consume negligible power because of the near-zero bias current. However, their bias voltage requirement of tens of volts necessitates redesigning the LoRa power delivery network \cite{reconfigurable-antenna}, as existing LoRa chips typically operate at only a few volts.}} After the synchronization process, all ports are observed cyclically. The benefit of utilizing the symbol-embedded FAS pilot is that it reduces the impact on system throughput and preserves the PHY frame structure of LoRa, compared to the continuous insertion of pilot symbols.\footnote{\textcolor{black}{One continuously inserted pilot symbol can be regarded as equivalent to $P/M$ redundancy symbols within $U$ transmitted symbols,} since it does not convey any modulation information. This leads to specific symbols within the LoRa frame for channel observation, which reduces the average throughput.}  Then the $u$-th symbol in one observation cycle is given by
	\begin{equation}\label{transmit-sig}
		\begin{split}
			&\tilde{x}_{m}[n]=\left\{\begin{aligned}
				&~x_{\rm{p}}\left[n + (u-1)\frac{P}{U}\right],\!\!\!\!\!&~~ 0\leq n<\frac{P}{U}-1,\\
				&~x_{m}[n], &~~ \frac{P}{U}\leq n\leq M\!-1,
			\end{aligned}\right.
		\end{split}
	\end{equation}
	\textcolor{black}{where $u\in\mathcal{U}=\{1,2,\dots,U\}$.} The received signal at the $l$-th FAS port is given by
	\begin{equation}\label{eq:rmn}
		\begin{split}
			r_m[n]=h_l\sqrt{E_{\rm{s}}}\tilde{x}_m[n]+z_l[n],
		\end{split}
	\end{equation}
	\textcolor{black}{where $l\in\mathcal{L}=\{1,2,\cdots,L\}$.} Moreover, $\color{black}{h_l}$ is the channel corresponding to the $l$-th FAS port and $z_l$ is the AWGN with $z_l{\color{black}{[n]}}\sim\mathcal{CN}\left(0,\frac{N_0}{2}\right)$, and $N_0$ is the single-side noise power spectral density. \textcolor{black}{Besides, $E_\text{s}$ is the transmitted power for one LoRa symbol and we have $E_{\rm{s}}=\sum_{n=0}^{M-1}|\sqrt{E_{\rm{s}}}x_m[n]|^2=\sum_{n=0}^{M-1}|\sqrt{E_{\rm{s}}}\tilde{x}_m[n]|^2$.} Generally, the signal-to-noise ratio (SNR) for LoRa communication is defined as \cite{bibitem8}
	\begin{equation}
		\Gamma\triangleq\frac{E_{\rm{s}}}{T_{\rm{s}} N_0 \Delta}=\frac{E_{\rm{s}}}{N_0 {\color{black}M}}.
	\end{equation}
	\textcolor{black}{Given the transmitted information $m$, the frequency of the continuous LoRa signal increases linearly from $\frac{m\Delta}{M}$ to $\frac{\Delta}{2}$, subsequently drops to $-\frac{\Delta}{2}$, and then increases linearly back to $\frac{m\Delta}{M}$. This characteristic spreads $E_\text{s}$ across the entire BW.}
	\subsection{Channel}\label{subsec:channel}
	Without loss of generality, we assume an $L$-port FAS with a \textcolor{black}{normalized} length of $W$ at the LoRa receiver, corresponding to the Rx-SISO-FAS model in \cite{bibitem28}. To model the channels \textcolor{black}{$\{h_l\}_{l=1}^L$} in \eqref{eq:rmn}, the block-correlation model in \cite{bibitem42} is adopted as it approaches the accuracy of the exact spatial correlation matrix generated by Clarke's model \cite{bibitem50}, \cite{bibitem51}. Specifically, the proposed model approximates Clarke’s model by employing a block diagonal matrix that fits the eigenvalues of the exact correlation matrix.\footnote{\textcolor{black}{The Clarke's model and block-correlation matrix are presented in \eqref{eq:sigma-sinc} and \eqref{out-corr}, respectively, in Appendix \ref{app:a}.}} Thus, the channel at the $l_b$-th port in block $b$ is \textcolor{black}{defined as}
	\begin{equation}\label{eq:h-channel}
		\begin{split}
			h_l\triangleq h_{b,l_b}=&\left(\sqrt{1-\mu^2} x_{b,l_b}+\mu x_{{b,l_0}}\right)\\
			&+j\left(\sqrt{1-\mu^2}y_{{b,l_b}}+\mu y_{{b,l_0}}\right),
		\end{split}
	\end{equation}
	\textcolor{black}{for $b\in\mathcal{B} =\{1,2,\dots,B\}, l_b\in\mathcal{L}_b=\{1,2,\dots,L_b\}$.} \textcolor{black}{$B$} is the number of blocks and $L_b$ is the dimension of the $b$-th block. Each $l$ has a one-to-one relationship with $(b,{\color{black}l_b})$. Then we have $L=\sum_{b=1}^{L_B}{L_b}$. \textcolor{black}{In addition,} $x_{0,l_0}$, $y_{0,l_0}$, $x_{b,l_b}$ and $y_{b,l_b}\sim\mathcal{N}\left(0,\frac{1}{2}\right)$ are independent and identically distributed (i.i.d.) random variables, and $\mu^2$ is the correlation parameter. The values of $B$ and $L_b$ are calculated by \cite[Algorithm I]{bibitem42}. \textcolor{black}{In \cite{bibitem42}, the parameter $\mu^2$ is chosen as a fixed value within the interval $(0.95,0.99)$. However, when $L$ exceeds a certain threshold, the fitting accuracy degrades. In this work, we propose a dynamic and judicious selection method of $\mu^2$ under different values of $L$. The rationale and procedure are detailed in Appendix \ref{app:a}.} Then the correlation coefficient between any $h_{i,c_i}$ and $h_{j,d_j}$ is given by 
	\begin{equation}
		\begin{split}
			{\rm{Corr}}\left\{h_{i,c_i},h_{j,d_j}\right\}&=\frac{{\rm{E}}\left\{\left|h_{j,d_j}\right|^2\right\}-{\rm{E}}\left\{\!h_{i,c_i}\right\}{\rm{E}}\left\{h_{j,d_j}^*\right\}}{{\rm{Var}}\left\{h_{i,c_i}\!\right\}{\rm{Var}}\left\{h_{j,d_j}\right\}}\\
			&=\left\{\begin{aligned}
				&1,& &i=j \text{ and } c=d,\\
				&\mu^2,& &i\neq j \text{ and } c=d,\\
				&0,& &c\neq d,
			\end{aligned}\right.
		\end{split}
	\end{equation}
	which indicates that the channels within the same block exhibit constant correlation with $\mu^2$, while the channels between different blocks are mutually independent.
	\subsection{Receiver}\label{sec:model-receiver}
	At the receiver, the pilot sequences within each symbol are \textcolor{black}{extracted and reconstructed} to recover the complete pilot symbols, \textcolor{black}{which is the inverse process from \eqref{pilot} to \eqref{transmit-sig}.} Assuming the flat-fading within $L\cdot U$ symbols, as the length of pilot symbols increases, the channel estimation error decreases.\footnote{\textcolor{black}{The probability that the least squares estimation error for the $l$-th FAS port is less than a given threshold $\epsilon$ can be expressed as $\mathrm{Pr}\big(\big|\tilde{h}_l-h_l\big|<\epsilon\big)=1-e^{-\frac{P\epsilon^2}{\Gamma}}$, where $\tilde{h}_l=\frac{\sum_{n_\mathrm{p}=0}^{P-1}(r_{\mathrm{p},l}[n_\mathrm{p}])^*x_\mathrm{p}[n_\mathrm{p}]}{\sum_{n_\mathrm{p}=0}^{P-1}(x_\mathrm{p}[n_\mathrm{p}])^*x_\mathrm{p}[n_\mathrm{p}]}$ is the estimated channel for the $l$-th port and $r_{\mathrm{p},l}[n]$ is the reconstructed pilot symbol from the received symbols.}} \textcolor{black}{Then the FAS channel at each sample point for both observation (at the $l$-th port) and transmission (at the selected port), excluding the initial observation round, is given by}
	\begin{equation}\label{eq:channel-h}
		\begin{split}
			h_{\rm FAS}[n]=\left\{\begin{aligned}
				&h_{l},\!&\!&\!0\leq n \leq\frac{\color{black}P}{U}-1,\\
				&\color{black}{\hat{h}},\!&\!&\!\frac{\color{black}P}{U}\leq n\leq {\color{black}M}-1,
			\end{aligned}\right.
		\end{split}
	\end{equation}
	\textcolor{black}{where $\hat{h}=h_{\hat{l}}$ and $\hat{l}=\arg\max\limits_{l\in\mathcal{L}}\left\{\left|h_1\right|,\dots,\left|h_L\right|\right\}$.} Then substi-tuting $h_{\rm FAS}$ for $h_l$ in \eqref{eq:rmn}, the received signal is expressed as
	\begin{equation}
		r_m[n]=h_{\rm FAS}[n]\tilde{x}_m[n]+z_{\rm FAS}[n],
	\end{equation}
	where $z_{\rm FAS}[n]\sim\mathcal{CN}\left(0,\frac{N_0}{2}\right)$ is the AWGN at the receiver. 
	
	The initial \textcolor{black}{$\frac{P}{U}$} samples contain no modulation information and do not require demodulation. Therefore, the truncated reference chirp signal at the receiver can be represented as
	\begin{equation}
		\tilde{x}_0[n]=\left\{\begin{aligned}
			&0,\!&\!&\!0\!\leq n<\frac{\color{black}P}{U}-1,\\
			&x_0[n],\!&\!&\!\frac{\color{black}P}{U}\!\leq\! n\leq {\color{black}M}-1.
		\end{aligned}\right.
	\end{equation}
	The \textcolor{black}{dechirped} signal can be acquired by multiplying $r_m$ and the complex conjugate of $\tilde{x}_0$, which is expressed as
	\begin{equation}
		\begin{split}
			\tilde{r}_m\![n]&=r_m[n]\tilde{x}_0^*[n]\\
			&=\left\{\begin{aligned}
				&0,\!&\!&\!0\!\leq n<\frac{\color{black}P}{U}-\!1,\\
				&\frac{\hat{h}}{M}e^{j2\pi\frac{mn}{M}}+\phi[n],\!&\!&\!\frac{\color{black}P}{U}\!\leq\! n\!\leq\! {\color{black}M}-1,
			\end{aligned}\right.
		\end{split}
	\end{equation}
	in which we have $\phi[n]=z_{\rm FAS}[n]\tilde{x}_0^*[n]$. Then a ${\color{black}M}$-point discrete Fourier transform (DFT) is applied to the down-chirp signal $\tilde{r}_m$. To maintain noise independence and simplify the analysis, zero samples are substituted with equal-length i.i.d.~noise samples, i.e., $\tilde{r}_m =\phi[n]$, when $0\leq n<\frac{P}{U}-1$. Then the DFT output signal can be expressed as 
	\begin{equation}\label{w-m}
		\begin{split}
			w_m[k]&=\sum_{n=0}^{{\color{black}M}-1}\tilde{r}_m[n]{e}^{-j2\pi\frac{kn}{{\color{black}M}}}\\
			&=\left\{\begin{aligned}
				&\underbrace{\left(1-\frac{\color{black}P}{{\color{black}M}\cdot U}\right)}_{\color{black}D}\hat{h}+\tilde{\phi}[k],\!&\!&\!\!k=m,\\
				&{\color{black}\xi_{m-k}}\hat{h}+\tilde{\phi}[k],\!&\!&\!\!k\neq{m},
			\end{aligned}\right.
		\end{split}
	\end{equation}
	\textcolor{black}{where, $k\in\mathcal{M}$ is the index of the frequency bin. Moreover, the proportion of pilot within a symbol is $1-D$. Note that this method is equivalent to directly calculating the cross-correlation of the received signal and each possible symbol, which coherently accumulates the channel gain across the samples within one symbol duration.} The noise at each bin follows same distribution as $\phi[k]$, i.e., $\tilde{\phi}[k]=\sum^{M-1}_{n=0}e^{\frac{-2\pi kn}{M}}\phi[n]\sim\mathcal{CN}\left(0,\frac{N_0}{2}\right)$. Moreover, $\xi_{m-k}$ is found using the exponential sum formula as
	\begin{equation}\label{eq:expsum}
		\begin{split}
			\xi_{m-k}&={\color{black}\frac{1}{\color{black}M}}\sum^{{\color{black}M}\!-\!1}_{n={\color{black}P}/{U}}e^{j2\pi\frac{(m-k)n}{\color{black}M}}\\
			&=\underbrace{{\color{black}\frac{1}{\color{black}M}}\frac{\sin\left(\frac{\pi \left(m-k\right){\color{black}P}}{{\color{black}M}\cdot U}\right)}{\sin\left(\frac{\pi (m-k)}{{\color{black}M}}\right)}}_{\color{black}{\left|\xi_{m-k}\right|}}\underbrace{e^{j\pi \left((m-k)\left(\frac{\color{black}P}{{\color{black}M}\cdot U}-\frac{1}{\color{black}M}\right)+1\right)}}_{\color{black}{\varphi\left(\xi_{m-k}\right)}}\\
		\end{split}
	\end{equation}
	\textcolor{black}{Note that when the dechirp operation is performed over the entire symbol, we have $w_m[k]=\hat{h}+\tilde{\phi}[k]$ for $k=m$, and $w_m[k]=\tilde{\phi}[k]$ for $k\neq m$ in \eqref{w-m}. Therefore, performing dechirp on the truncated symbol disperses a portion of the symbol power into bins indexed by $k\neq m$. This effect can be regarded as the `leakage interference' during the symbol detection.} 
	
	\textcolor{black}{Both non-coherent and coherent symbol detection schemes are considered, which are given as follows.}
	\subsubsection{\textcolor{black}{Non-coherent demodulation}}\textcolor{black}{This approach aims to directly select the index $k$ corresponding to the largest magnitude of $w_m[k]$, where the magnitude of the $w_m[k]$ is given by}
	\begin{equation}\label{abs-w}
		\begin{split}
			\big|w_m[k]\big|
			&\color{black}=\left\{\begin{aligned}
				&\color{black}\left|{\color{black}D}\hat{h}+\tilde{\phi}[k]\right|,&\color{black}k=m,\\
				&\color{black}\left|\xi_{m-k}\hat{h}+\tilde{\phi}[k]\right|,&\color{black}k\neq{m},
			\end{aligned}\right.\\
			&=\left\{\begin{aligned}
				&\left|{\color{black}D}\big|\hat{h}\big|+\tilde{\phi}_1[k]\right|,&k=m,\\
				&\left|\big|\xi_{m-k}\big|\big|\hat{h}\big|+\tilde{\phi}_2[k]\right|,&k\neq{m},
			\end{aligned}\right.
		\end{split}
	\end{equation}
	\textcolor{black}{and we have $\tilde{\phi}_{1}[k]\!=\!\big|\tilde{\phi}[k]\big|e^{j\left(\varphi(\tilde{\phi}[k])-\varphi(\hat{h})\right)}\sim\mathcal{CN}(0,N_0)$, $\tilde{\phi}_{2}[k]=\big|\tilde{\phi}[k]\big|e^{j\big(\varphi(\tilde{\phi}[k])-\varphi(\hat{h})-\varphi(\xi_{m-k})\big)}\sim\mathcal{CN}(0,N_0)$. $\tilde{\phi}_{1}[k]$ and $\tilde{\phi}_{2}[k]$ have different values of $k$, indicating their mutual independence. Then the demodulated signal $\tilde{m}$ is obtained as}
	\begin{equation}\label{eq:mopt_ncoh}	\tilde{m}^{\rm{(ncoh)}}=\arg\max\limits_{k\in\mathcal{M}}\left|w_m[k]\right|,
	\end{equation}
	\subsubsection{\textcolor{black}{Coherent demodulation}} \textcolor{black}{Since the CSI is available in the proposed system, the coherent receiver can employ the CSI to compensate for the channel phase in $w_m[k]$, and subsequently extracts the index $k$ corresponding to the maximum real part, which is given by}
	\begin{equation}\label{abs-w-coh}\color{black}
		\begin{split}
			&\mathrm{Re}\left(w_m[k]e^{-j\phi\left(\hat{h}\right)}\right)\\
			&\,=\left\{\begin{aligned}
				&D\big|\hat{h}\big|+\mathrm{Re}\left(\tilde{\phi}[k]e^{-j\varphi\left(\hat{h}\right)}\right),&k=m,\\
				&\mathrm{Re}\left(\xi_{m-k}\right)\big|\hat{h}\big|+\mathrm{Re}\left(\tilde{\phi}[k]e^{-j\varphi\left(\hat{h}\right)}\right),&k\neq{m}.
			\end{aligned}\right.
		\end{split}
	\end{equation}
	\textcolor{black}{and}
	\begin{equation}\label{eq:mopt_coh}\color{black}
		\tilde{m}^{(\rm{coh})}=\arg\max\limits_{k\in{\color{black}\mathcal{M}}}\mathrm{Re}\left(w_m[k]e^{-j\varphi\left(\hat{h}\right)}\right),
	\end{equation}
	\textcolor{black}{where $\mathrm{Re}\left(\tilde{\phi}[k]e^{-j\varphi\left(\hat{h}\right)}\right)\sim\mathcal{N}\left(0,\frac{N_0}{2}\right)$.} 
	\textcolor{black}{When the pilot samples occupy half of the symbol duration, i.e. $1-D=\frac{1}{2}$, the real part of $\xi_{m-k}$ is given by $\mathrm{Re}\left(\xi_{m-k}\right)=-|\xi_{m-k}| \cos\left(\pi\frac{m-k}{2}-\pi\frac{m-k}{M}\right)=\left\{\begin{aligned}
			&0,\!&\!&\!\!m-k\text{ is even},\\
			&-\tfrac{1}{M},\!&\!&\!\!m-k\text{ is odd}
		\end{aligned}\right.$. Since $M$ is large, the leakage interference becomes negligible, and the detection error is dominated by noise. This observation highlights the potential of employing the coherent detection scheme, when the pilot proportion is large.} 
	\subsection{Performance Metric}
	According to \eqref{eq:mopt_ncoh} and \eqref{eq:mopt_coh}, note $\tilde{m}\in\{\tilde{m}^{(\rm{ncoh})},\tilde{m}^{(\rm{coh})}\}$, a symbol error occurs when $\tilde{m}\neq{m}$, \textcolor{black}{as the largest value of bins indexed by $k\in\mathcal{K}$ with $k\neq m$ exceeds the $m$-th bin.} Therefore, the SER of non-coherent and coherent receivers can be expressed by 
	\begin{equation}\label{eq:ser_ncoh}
		\begin{split}
			{P}_{\rm{s}}^{(\rm{noch})}&=\mathrm{Pr}\left(\begin{array}{c}
				\max\limits_{k\in\mathcal{K},k\neq{m}}\left|\left|\xi_{m-k}\right|\big|\hat{h}\big|\!+\!\tilde{\phi}_2[k]\right|\\
				\geq\left|{\color{black}D}\big|\hat{h}\big|+\tilde{\phi}_1[m]\right|
			\end{array}\!\right)\\
			&\triangleq{\rm{Pr}}\left(\max\limits_{k\in\mathcal{K}, k\neq{m}}X_{k}^{(\rm{noch})}\geq{X}_{m}^{(\rm{noch})}\right),
		\end{split}
	\end{equation}
	
	\begin{equation}\label{eq:ser_coh}\color{black}
		\begin{split}
			{P}_{\rm{s}}^{(\rm{coh})}&=\mathrm{Pr}\left(\!\!\begin{array}{c}
				\max\limits_{k\in\mathcal{K},k\neq{m}}\mathrm{Re}\left(\xi_{m-k}\right)\big|\hat{h}\big|+\mathrm{Re}\left(\tilde{\phi}[k]e^{-j\varphi\left(\hat{h}\right)}\right)\\
				\geq{D}\big|\hat{h}\big|+\mathrm{Re}\left(\tilde{\phi}[k]e^{-j\varphi\left(\hat{h}\right)}\right)
			\end{array}\!\!\right)\\
			&\triangleq{\rm{Pr}}\left(\max\limits_{k\in\mathcal{K}, k\neq{m}}X_{k}^{(\rm{coh})}\geq{X}_{m}^{(\rm{coh})}\right),
		\end{split}
	\end{equation}
	respectively. \textcolor{black}{Note $P_\text{s}\in\{P_{\rm{s}}^{(\rm{ncoh})},P_{\rm{s}}^{(\rm{coh})}\}$.  The average SER is given by  ${P}_{{\rm{s}}}=\frac{1}{M}\sum\limits^{M-1}_{m=0}{P}_{{\rm{s}}|m}={P}_{\rm{s}|m}$, $\forall m\in\mathcal{M}$, where $P_{\rm{s}|m}$ represents the conditional SER when the transmitted symbol is $m$. This is derived under the assumption that all symbols are transmitted with equal probability. Furthermore, for large $M$, the BER of the proposed system $P_{\rm{b}}\in\{P_{\rm{b}}^{(\rm{ncoh})},P_{\rm{b}}^{(\rm{coh})}\}$ is given by ${P}_{\rm{b}}=\frac{M/2}{M-1}{P}_{\rm{s}}\approx\frac{1}{2} {P}_{\rm{s}}$ \cite{bibitem8}. Moreover, the system throughput $R\in\{R^{(\rm{ncoh})},R^{(\rm{coh})}\}$ is defined as $R=\frac{\mathrm{log}_2({\color{black}M})}{T_{\rm{s}}}\left(1-P_{\rm{s}}\right)$, which represents the number of correctly detected bits per unit time. As both $P_{\rm{b}}$ and $R$ are \textcolor{black}{linear} functions of $P_{\rm{s}}$, this paper focuses on SER performance.}
	
	\section{Performance Analysis}\label{sec:analysis}
	This section provides our main analytical results for the proposed LoRa-FAS. \textcolor{black}{For analytical simplicity, perfect CSI is assumed, which can be regarded as the upper bound of the practical LoRa-FAS.} Given the complexity of determining the exact SER under correlated channel conditions for FAS, we resort to deriving a closed-form approximation for the SER, incorporating the effects of the embedded pilot overhead. 
	Specifically, in \eqref{eq:ser_ncoh} and \eqref{eq:ser_coh}, $X_k\in\{X_k^{(\rm{ncoh})},X_k^{(\rm{coh})}\}$ and $X_m\in\{X_m^{(\rm{ncoh})},X_m^{(\rm{coh})}\}$ are not independent because of the common $|\hat{h}|$. Additionally, $\max\limits_{k\in\mathcal{K},k\neq{m}}X_k$ includes double maximum operations, where the inner maximum selects the FAS port corresponding to the largest magnitude of the channel, while the outer maximum is for detecting the LoRa symbol. \textcolor{black}{Here, we first derive the approximate PDF and CDF of $|\hat{h}|$ and present them in a more manageable form with the help of the Lambert $W$-function. Then the closed-form SER representations for a given $|\hat{h}|$ are derived under both non-coherent and coherent detection schemes.} Finally, the closed-form approximation of SER for the proposed LoRa-FAS is obtained by averaging, i.e., integrating the conditional SER on the PDF.
	\subsection{\textcolor{black}{Approximate PDF and CDF of FAS Channel}}\label{sec:3a}
	\textcolor{black}{The correlation between ports plays a crucial role in characterizing the potential of the sub-half-wavelength scale of FAS. The block-correlation model was first introduced in the performance analysis of FAMA \cite{bibitem33} to approximate Jake’s or Clarke’s correlation model. However, directly applying it to LoRa-FAS poses challenges in deriving closed-form performance expressions. To our knowledge, an exact and tractable statistical expression of the block-correlation model for the Rx-SISO-FAS channel remains an open issue. In this work, we derived a concise and manageable expression for $|\hat{h}|$. The simulation results confirm that the proposed expression achieves satisfactory accuracy.} 
	
	Now, we first focus on a single-channel block characterized by a constant correlation factor. For the $b$-th channel block in \eqref{eq:channel-h}, note the largest magnitude within $L_b$ ports as $|\hat{h}_b|$, i.e., $\big|\hat{h}_b\big|=\max\limits_{l_b\in\mathcal{L}_b}{\big|h_{b,l_b}\big|}$. Then we have the following lemma.
	
	\begin{lemma}\label{lemma:1}\textcolor{black}{The CDF of $\big|\hat{h}_b\big|$ corresponding to the $b$-th block in the block-correlation model can be approximated as}
		\begin{equation}\label{eq:F-constcorr}
			\begin{split}
				F_{\left|\hat{h}_b\right|}\left(r\right)\approx&\,\left(1-e^{-\frac{1}{\mu^2}\left(r-\delta_{b}\right)^2}\right){H\big(r-\delta_{b}\big)},
			\end{split}
		\end{equation}
		where $\delta_{b}$ is the shift parameter associated with the size of the $b$-th channel block, expressed as
		\begin{equation}\label{eq:Wfun}
			\delta_{b}={\sqrt{\frac{1-\mu^2}{2}{\cal W}\left(\frac{\left(L_b-2\right)^2}{2\pi}\right)}}.
		\end{equation}
	\end{lemma}
	\begin{proof}
		See Appendix \ref{app:b}.
	\end{proof}
	\textcolor{black}{Lemma \ref{lemma:1} shows the behavior of a single block. Since $\mu^2$ is close to $1$, \eqref{eq:F-constcorr} equals to a shift of the CDF of Rayleigh channel, following $\mathrm{Ra}\left(\frac{\sqrt{2}}{2}\right)$. Notice this shift can be regarded as the potential gain achievable at sub-half-wavelength spacing. More specifically, \cite{bibitem42} and \cite{bibitem43} demonstrate that, for large $L$, the number of blocks is approximately $2W$, i.e., the number of half-wavelength, which aligns well with the spatial sampling theory. Thus, the additional gain provided by sub-half-wavelength spacing inside FAS can be expressed by the shift terms of all single blocks.}
	
	Based on Lemma \ref{lemma:1}, we extend the result to multiple blocks.
	
	\begin{proposition}\label{lemma:cdf-allblocks}
		The CDF of $|\hat{h}|$ among multiple blocks in the block-correlation model can be approximated as
		\begin{equation}\label{eq:cdf-allblocks}
			\begin{split}
				F_{|\hat{h}|}\left(r\right)\approx&\,\frac{1}{B}\sum^{B}_{b=1}\left[1-e^{-\frac{1}{\mu^2}\left(r-\delta_{b}\right)^2}\right]^{B}\color{black}{ H\big(r-\hat{\delta}\big)},
			\end{split}
		\end{equation}
		where $\hat{\delta}=\max\limits_{b\in\mathcal{B}}\delta_{b}$ denotes the maximum shift among all channel blocks.
	\end{proposition}
	
	\begin{proof}
		See Appendix \ref{app:c}.
	\end{proof}
	
	Notice, the CDF presented in Proposition \ref{lemma:cdf-allblocks} is more applicable to Clarke's model than Jake's model, where the eigenvalues are predominantly characterized by several equal maximum eigenvalues, as described in the following lemma. 
	
	\begin{lemma}\label{lemma:largeL}
		For large $L$, the correlation matrix under Clarke's model ($\mathbf{\Sigma}$ in \eqref{eq:sigma-sinc}) exhibits roughly $\frac{2W\cdot L}{L-1}\approx2W$ equal eigenvalues, each with value approximately $\frac{L-1}{2W}$.
	\end{lemma}
	
	\begin{proof}
		See \cite[Lemma 1 \& Corollary 1]{bibitem42}.
	\end{proof}
	
	\textcolor{black}{Besides, although the CDF given by \eqref{exact-cdf-blk} in Appendix \ref{app:c} is precise, differentiating it does not yield a tractable PDF suitable for further analysis. In contrast, the CDF in \eqref{eq:cdf-allblocks} exhibits good accuracy and analytical tractability. Accordingly, the approximate PDF is given as follows.} 
	
	\begin{figure*}[!t]\label{Sim1}\color{black}
		\centering
		\subfloat[]{\includegraphics[width=0.3333\linewidth]{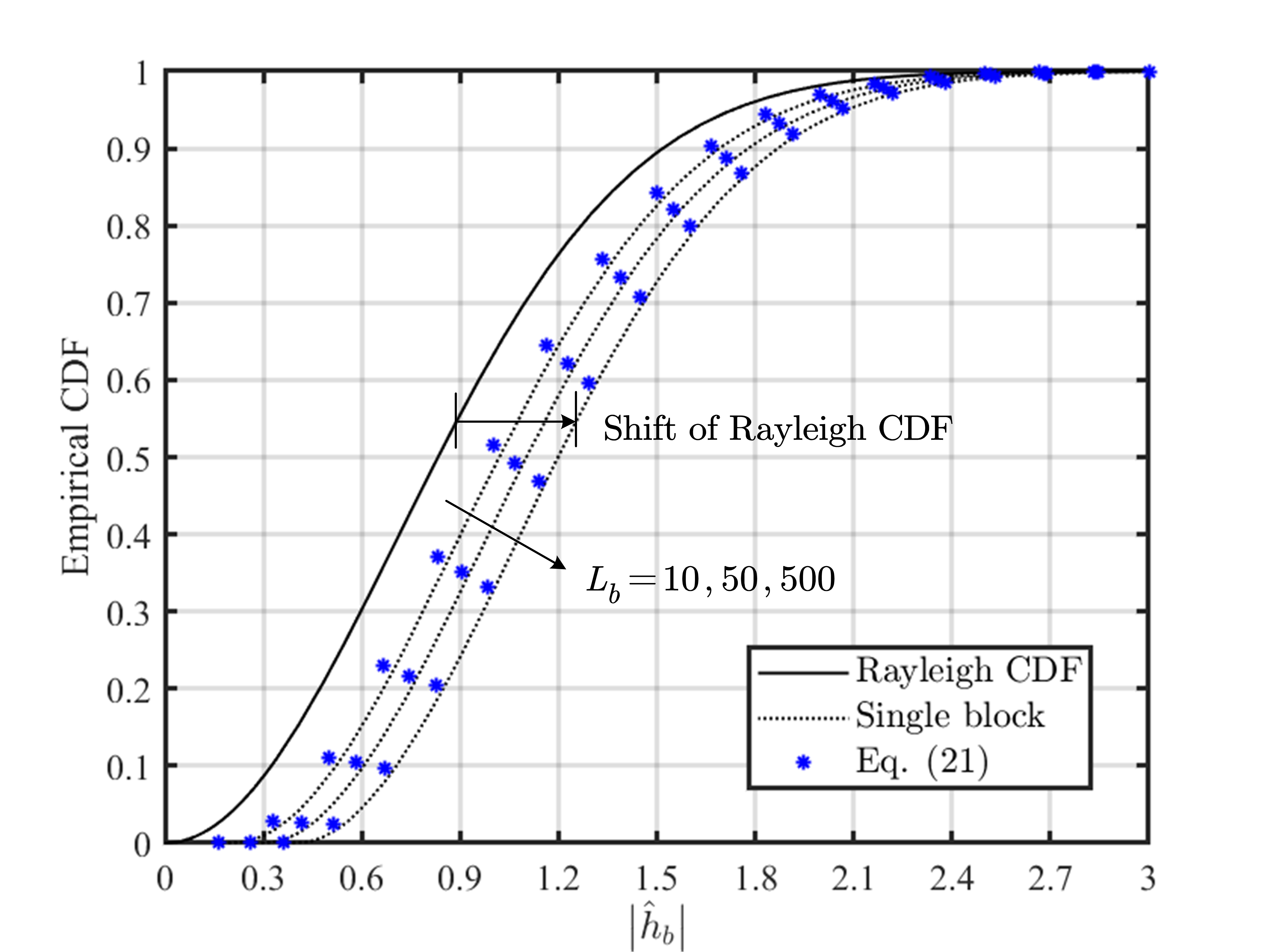}%
			\label{sim:1a}}
		\hfil
		\subfloat[]{\includegraphics[width=0.3333\linewidth]{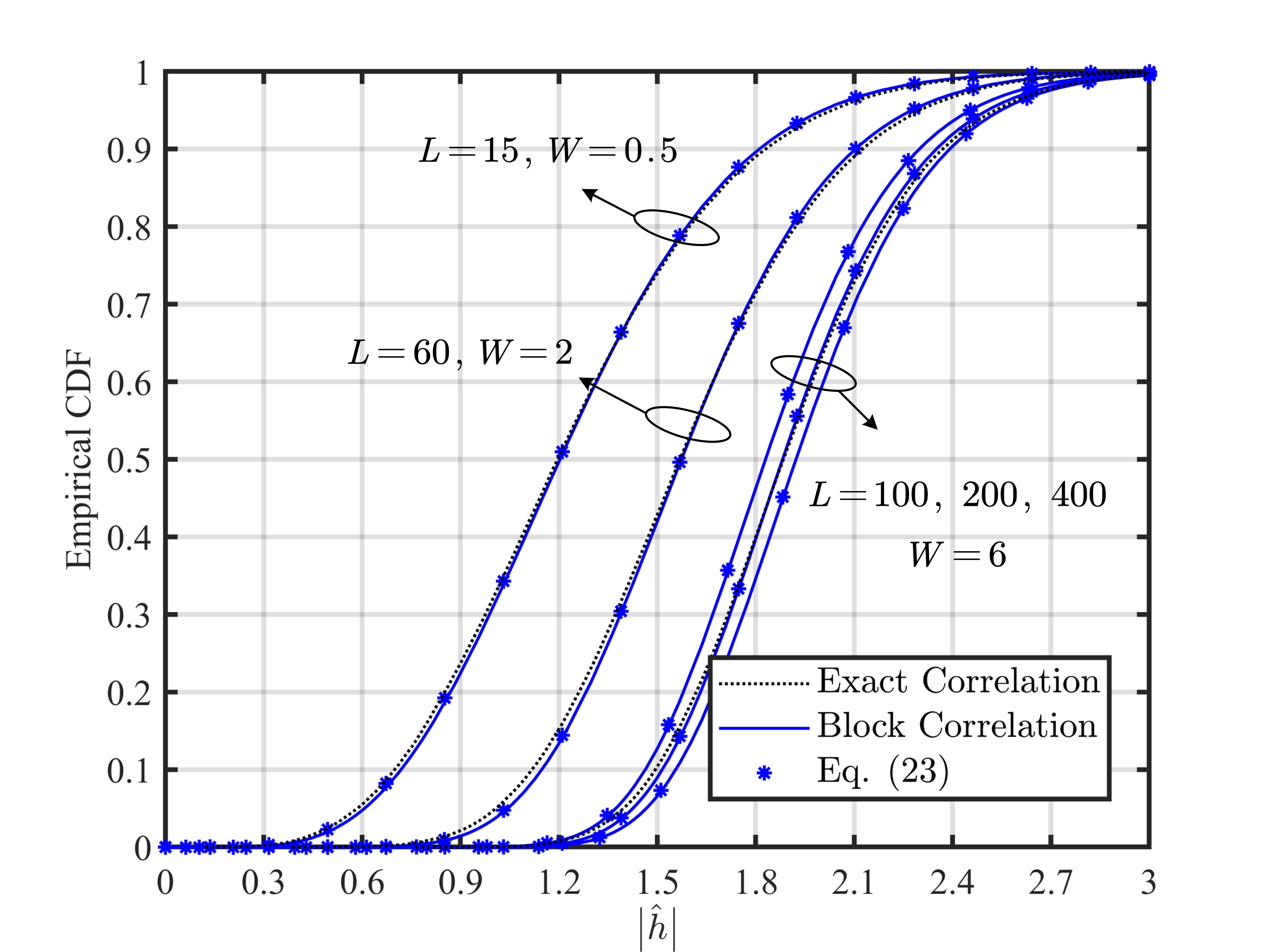}%
			\label{sim:1b}}
		\hfil
		\subfloat[]{\includegraphics[width=0.3333\linewidth]{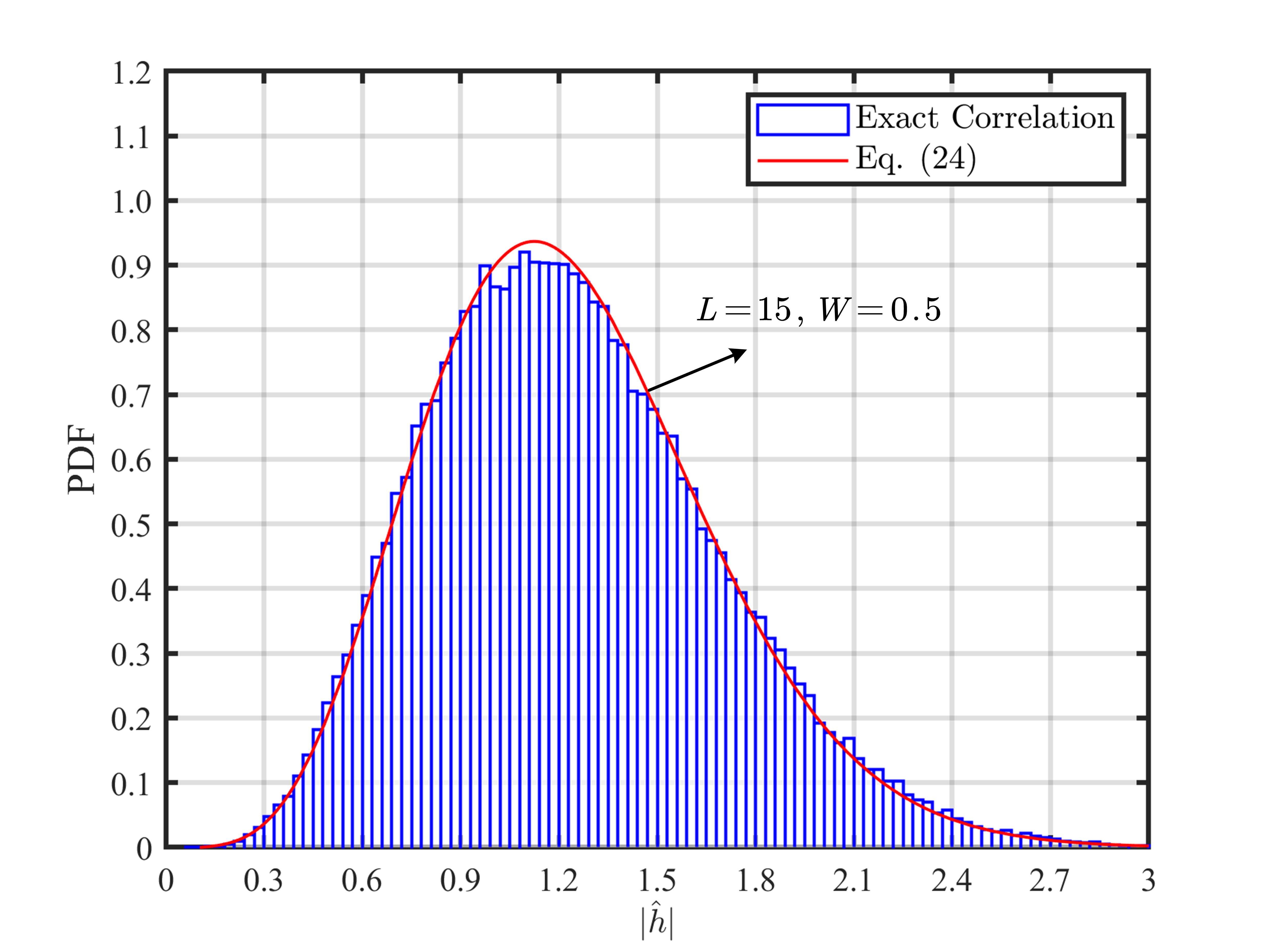}%
			\label{sim:1c}}
		\caption{Numerical results of the CDFs and PDFs for the analytical approximation and Clarke's exact models. (a) The CDF results of $\big|\hat{h}_b\big|$ with $\mu^2=0.97$, and $\delta_b=\{0.1624,0.2574,0.3560\}$ for $L_b=\{10,50,500\}$, corresponding to Lemma. 1. (b) The CDF results of $\big|\hat{h}\big|$ with $\mu^2=0.97$, corresponding to Proposition. 1. (c) The PDF results of $\big|\hat{h}_b\big|$ with $\mu^2=0.97$, corresponding to Proposition. 2. In these results, lines correspond to Monte-Carlo simulations, while markers specify the results from the analytical expressions.}
		\label{Fig.2}
	\end{figure*}
	
	\begin{proposition}
		The PDF of $|\hat{h}|$ among multiple blocks in the block-correlation model can be approximated by
		\begin{equation}\label{eq:f-hstar}
			\begin{split}
				f_{|\hat{h}|}(r)\approx&\sum^{B}_{b=1}\,\sum^{B}_{\tilde{b}=1}(-1)^{\tilde{b}-1}\,\binom{B-\!1}{\tilde{b}-\!1}e^{-\frac{\tilde{b}}{\mu^2}\left(r-\delta_{b}\right)^2}\\
				&\times{\frac{2}{\mu^2}\left(r-\delta_{b}\right)\color{black}{H\big(r-\hat{\delta}\big)}}\\
			\end{split}
		\end{equation}
		\textcolor{black}{where $\binom{B-1}{\tilde{b}-1}$ is the binomial coefficient.}
	\end{proposition}
	
	\begin{proof}
		To derive \eqref{eq:f-hstar}, we first take the derivative of \eqref{eq:cdf-allblocks} and then apply the binomial expansion on it. \textcolor{black}{Strictly, to ensure that the integral of the PDF over its domain equals one, a Dirac delta function located at the lower bound of its domain, i.e., $F_{|\hat{h}|}(\hat{\delta})D(r-\hat{\delta})$, should be added in the PDF. Here, we directly drop the delta function term as its value is extremely small, closing to 0.}  
	\end{proof}
	Note that \eqref{eq:f-hstar} is relatively manageable since it consists of a summation of products involving only exponential and power functions. \textcolor{black}{Fig. \ref{Fig.2} shows the numerical results corresponding to the key findings presented in Section \ref{sec:3a}. These results demonstrate that the proposed analytical approximation fits well with the Clark's model. Besides, in Fig. 2(b) for the case $W=6$, as $L$ increases, the block-correlation model and analytical results first converge toward the exact model and then gradually deviate from it. This behavior is consistent with the description in Section \ref{subsec:channel} and has been explained in Appendix \ref{app:a}.}
	
	\subsection{\textcolor{black}{SER Analysis of Non-coherent Receiver}}
	\textcolor{black}{In this subsection, we aim to analyze the error performance of the proposed LoRa-FAS system under a non-coherent receiver. Given $|\hat{h}|$, we first derive a closed-form expression for the conditional SER. The final result is then obtained by averaging over it over the PDF of $|\hat{h}|$. Specifically, based on \eqref{eq:ser_ncoh}, the SER under non-coherent detection can be expressed as}
	\begin{equation}\label{eq:ser2}\color{black}
		\begin{split}
			&{P}_{\rm{s}}^{(\rm{ncoh})}\\
			&\quad={\rm{E}}_{|\hat{h}|}\left\{1\!-\!{\rm{Pr}}\left(\max\limits_{k\in\mathcal{K}, k\neq{m}}\!X^{(\rm{ncoh})}_{k|\,|\hat{h}|=r}\leq{X}^{(\rm{ncoh})}_{m|\,|\hat{h}|=r}\right)\right\}\!\!\\
			&\quad={\rm{E}}_{|\hat{h}|}\left\{{\rm{E}}_{X_{m|\,|\hat{h}|=r}^{(\rm{ncoh})}}\left\{1-F_{\hat{X}^{(\rm{ncoh})}_{|\,|\hat{h}|=r}}\left(x^{(\rm{ncoh})}_{m|\,|\hat{h}|=r}\right)\right\}\right\}\\
			&\quad\triangleq{\rm{E}}_{|\hat{h}|}\left\{P_{\text{s}|\,|\hat{h}|=r}^{(\rm{ncoh})}\right\}, \forall m\in\mathcal{M},
		\end{split}
	\end{equation}
	\textcolor{black}{where $P_{\text{s}|\,|\hat{h}|=r}^{(\rm{ncoh})}$ denotes the conditional SER under the given channel magnitude. Let $\hat{X}_{|\,|\hat{h}|=r}^{(\rm{ncoh})}=\max\limits_{k\in\mathcal{K},k\neq{m}}X_{k|\,|\hat{h}|=r}^{(\rm{ncoh})}$ denote the maximum magnitude among the frequency bins indexed by $k\neq m$. $F_{\hat{X}_{|\,|\hat{h}|=r}^{(\rm{ncoh})}}(\cdot)$ is the conditional CDF of $\hat{X}_{|\,|\hat{h}|=r}^{(\rm{ncoh})}$.}
	
	\textcolor{black}{Recall that ${X}_{m|\,|\hat{h}|=r}^{(\rm{ncoh})}$ denotes the magnitude at the $m$-th frequency bin, which corresponds to the correct index of the transmitted symbol. Meanwhile, ${X}_{k|\,|\hat{h}|=r}^{(\rm{ncoh})}$ for $k\in\mathcal{K}$ and $k\neq m$ represent the values associated with the other bins, reflecting the leakage interference and noise components. If these two factors are neglected, then ${X}_{m|\,|\hat{h}|=r}^{(\rm{ncoh})}$ is always the largest value. However, due to the noise and leakage interference,  $\hat{X}_{|\,|\hat{h}|=r}^{(\rm{ncoh})}$ may become comparably large, potentially causing detection errors. To evaluate the SER, it is necessary to characterize the statistical distributions of both $\hat{X}_{|\,|\hat{h}|=r}^{(\rm{ncoh})}$ and ${X}_{m|\,|\hat{h}|=r}^{(\rm{ncoh})}$. Specifically, $\hat{X}_{|\,|\hat{h}|=r}^{(\rm{ncoh})}$ and ${X}_{m|\,|\hat{h}|=r}^{(\rm{ncoh})}$ are Rician distributed variables following $\mathrm{Ri}\Big(\left|\xi_{{\pm{1}}}\right|{r},\sqrt{\frac{N_0}{2}}\Big)$ and $\mathrm{Ri}\Big(D{r},\sqrt{\frac{N_0}{2}}\Big)$, respectively. Thus, the conditional SER can be obtained by averaging, i.e., integrating the PDF of ${X}_{m|\,|\hat{h}|=r}^{(\rm{ncoh})}$ over the CDF of $\hat{X}_{|\,|\hat{h}|=r}^{(\rm{ncoh})}$, which is given by} 
	\begin{multline}\label{eq:exact-ser}
		\color{black}{{P}^{(\rm{ncoh})}_{{\rm{s}}|\,|\hat{h}|=r}=1-\int^{\infty}_{0}F_{\hat{X}_{|\,|\hat{h}|=r}^{(\rm{ncoh})}}(z)f_{{X}_{m|\,|\hat{h}|=r}^{(\rm{ncoh})}}(z)dz}\\
		\color{black}{=1-\frac{2}{N_0}\int^{\infty}_{0}\prod_{k=1\atop k\neq{m}}^{M}\left[1-Q_1\left(\frac{\left|\xi_{m-k}\right|{r}}{\sqrt{\frac{N_0}{2}}},\frac{{z}}{\sqrt{\frac{N_0}{2}}}\right)\right]}\\
		\color{black}{\times{z}e^{-\frac{z^2+D^2r^2}{N_0}}I_0\left(\frac{2Dzr}{N_0}\right)dz.}
	\end{multline}
	\textcolor{black}{Noted that the integral cannot be expressed in closed form. In this subsection, we derive a closed-form approximation of  \eqref{eq:exact-ser}.}
	
	\textcolor{black}{First, we approximate the CDF of the largest interference conditioned on a given channel magnitude as follows.} 
	
	\begin{lemma}\label{lemma:con-ser-approx}
		The conditional CDF of $\hat{X}_{|\,|\hat{h}|=r}^{(\rm{ncoh})}$ can be approximated by
		\begin{subequations}\label{cdf-Xr}
			\begin{align}
				F_{{\hat{X}_{|\,|\hat{h}|=r}}^{(\rm{ncoh})}}(z)=&\prod_{k=1\atop k\neq{m}}^{M}\left[1-Q_1\left(\frac{\left|\xi_{m-k}\right|{r}}{\sqrt{\frac{N_0}{2}}},\frac{{z}}{\sqrt{\frac{N_0}{2}}}\right)\right]\label{con-cdf-x-exact}\\
				\approx&\left[\!1\!-2\,Q_1\!\left(\frac{\left|\xi_{\pm{1}}\right|r}{\sqrt{\frac{N_0}{2}}},\frac{{z}}{\sqrt{\frac{N_0}{2}}}\right)\right]\notag\\
				&{\color{black}\times{H}\left(z\!-\!\sigma_{M-3}^{(\rm{ncoh})}\,\right)},\label{con-cdf-x-approx}
			\end{align}
		\end{subequations}
		where $\sigma^{(\rm{ncoh})}_{M-3}$ is defined as \eqref{eq:sigma-M3-ncoh} in Appendix \ref{app:d}. Moreover, ${\pm{1}}$ means $k=m\pm1$.
	\end{lemma}
	
	\begin{proof}
		See Appendix \ref{app:d}.
	\end{proof}
	\textcolor{black}{Based on the proofs of Lemma \ref{lemma:con-ser-approx}, \eqref{con-cdf-x-approx} includes the influence of both the leakage interference caused by pilot overhead and the noise at the receiver, combining these factors into a single expression.} 
	
	Then we substitute \eqref{con-cdf-x-approx} into \eqref{eq:exact-ser}, and set the Heaviside step function on the lower bound of the integral, which is given by
	\begin{multline}\label{eq:con-ser}
		{P}^{(\rm{ncoh})}_{\text{s}|\,|\hat{h}|=r}\approx1- \frac{2}{N_0}\int^{\infty}_{\sigma_{M-3}^{(\rm{ncoh})}}\!\left[1-2Q_1\left(\frac{\big|\xi_{{\pm{1}}}\big|}{\sqrt{\frac{N_0}{2}}}r,\frac{z}{\sqrt{\frac{N_0}{2}}}\right)\right]\\
		\times{z}e^{-\frac{z^2+D^2r^2}{N_0}}I_0\left(\frac{2Dzr}{N_0}\right)dz.
	\end{multline}
	Note that when $\sigma_{M-3}^{(\rm{ncoh})}=0$, \eqref{eq:con-ser} can be solved by \cite[Theorem 1]{bibitem21}. However, evaluating the integral with a nonzero lower limit presents significant challenges. Therefore, we provide the following approximation for the Marcum Q-function.
	\begin{lemma}\label{lemma:marcum-q-approx}
		For $z\geq\left|\xi_{{\pm{1}}}\right|r$, the Marcum $Q$-function in \eqref{eq:con-ser} can be approximated by
		\begin{equation}\label{eq:Q1-approx}
			Q_1\left(\frac{\left|\xi_{{\pm{1}}}\right|{r}}{\sqrt{\frac{N_0}{2}}},\frac{{z}}{\sqrt{\frac{N_0}{2}}}\right)
			\approx{\color{black}\sum_{i=1}^2\frac{\alpha_i}{r^{i-1}}e^{-\frac{\left(z-\right|\xi_{{\pm{1}}}\left|r\right)^{2}}{N_0} }},
		\end{equation}
		\textcolor{black}{where $\alpha=\left[\frac{1}{4},\frac{1}{4\left|\xi_{{\pm{1}}}\right|}\right]$.}
	\end{lemma}
	\begin{proof}
		See Appendix \ref{app:e}.
	\end{proof}
	In LoRa systems, the number of $M$ is typically huge, leading to relatively small values of $\left|\xi_{{\pm 1}}\right|$, which is in accordance with the condition in Lemma \ref{lemma:marcum-q-approx}. 
	
	\textcolor{black}{Second, the Rician distributed $X^{(\rm{ncoh})}_{m|\,|\hat{h}|=r}$ can be approximated as a Gaussian distribution for large values of the Rician shape parameter, following $\mathcal{N}\Big(Dr,\sqrt{\frac{N_0}{2}}\Big)$. Based on this, the conditional SER can be rewritten as}
	\begin{multline}\label{eq:con-ser-approx}
		{P}_{{\rm{s}}|\,|\hat{h}|=r}^{(\rm{ncoh})}\approx1-\frac{1}{\sqrt{N_0\pi}}\int_{\sigma_{M-3}^{(\rm{ncoh})}}^{\infty}e^{-\frac{\left(z-Dr\right)^2}{N_0}}\\
		\times\left[1-2\left({\color{black}\sum_{i=1}^2\frac{\alpha_i}{r^{i-1}}z^{i-1}}e^{-\frac{\left(z-\left|\xi_{{\pm 1}}\right|r\right)^{2}}{N_0}}\right)\right]dz,
	\end{multline}
	Then the closed-form expression of \eqref{eq:con-ser-approx} is given as follows.
	
	\begin{lemma}\label{lemma:con-ser-approx-cl}
		The closed-form expression of the conditional SER $P_{{\rm{s}}|\,|\hat{h}|=r}^{(\rm{ncoh})}$ in \eqref{eq:exact-ser} can be approximated as \eqref{eq:con-ser-approx-cl}, at the top of the next page.
		\begin{figure*}[!h]
			\begin{equation}\label{eq:con-ser-approx-cl}
				\begin{split}
					{\color{black}{P}^{(\rm{ncoh})}_{{\rm{s}}|\,|\hat{h}|=r}}\approx1-&{\frac{1}{2}\,\underbrace{\mathrm{erfc}\left(\frac{{\sigma_{M-3}^{(\rm{ncoh})}}-Dr}{\sqrt{N_0}}\right)}_{\color{black}A_1^{(\rm{ncoh})}}}-\frac{1}{\sqrt{N_0\pi}}\sum_{i=1}^2{\alpha_{i}}\Bigg\{{\color{black}\frac{i-1}{r^{i-1}}\frac{N_0}{2}}\underbrace{e^{\frac{2\sigma_{M-3}^{(\rm{ncoh})}\left(\left|\xi_{{\pm{1}}}\right|^2+D\right)r-\left(\left|\xi_{{\pm{1}}}\!\right|^2+D^2\right){r^2}-{2}\left(\sigma_{M-3}^{(\rm{ncoh})}\right)^2}{N_0} }}_{\color{black}A_2^{(\rm{ncoh})}}\\
					&\quad+{\color{black}\sqrt{\frac{N_0\pi}{8}}\left(\frac{\left|\xi_{{\pm{1}}}\right|^2+D}{2}\right)^{i-1}}\underbrace{e^{-{\frac{\left(\left|\xi_{{\pm{1}}}\!\right|^2-D\right)^2r^2}{2N_0}}}\mathrm{erfc}\left(\frac{2\sigma_{M-3}^{(\rm{ncoh})}-\left(\left|\xi_{{\pm{1}}}\right|^2+D\right)r}{\sqrt{2N_0}}\right)}_{\color{black}A_3^{(\rm{ncoh})}}\Bigg\}
				\end{split}
			\end{equation}
			\hrule
		\end{figure*}
	\end{lemma}
	\begin{proof}
		\textcolor{black}{See Appendix \ref{app:f}}.
	\end{proof}
	
	Based on \eqref{eq:con-ser-approx-cl} and \eqref{eq:f-hstar}, the closed-form approximation of SER for the proposed LoRa-FAS can be obtained as follows.
	
	\begin{proposition}\label{prop:ser-ncoh}
		The closed-form SER of the proposed LoRa-FAS under non-coherent detection can be approximated by
		\begin{equation}\label{ser-ncoh}
			\begin{split}
				&{P}^{(\rm{ncoh})}_{\rm{s}}\approx1-\sum^{B}_{b=1}\sum^{B}_{\tilde{b}=1}\,(-1)^{\tilde{b}-1}\binom{B-1}{\tilde{b}-1}\\
				&{\color{black}\,\times\frac{1}{\mu^2}\Bigg[\bar{A}_1^{(\rm{ncoh})}-\sum_{i=1}^2\frac{{2\alpha_i}}{\sqrt{N_0\pi}}}\\
				&{\color{black}\times\left(\frac{N_0(i-1)}{2}\bar{A}_2^{(\rm{ncoh})}+\left(\frac{|\xi_{\pm1}|+D}{2}\right)^{i-1}\bar{A}_3^{(\rm{ncoh})}\right)\Bigg]}\\
			\end{split}
		\end{equation}
		\textcolor{black}{where $\bar{A}_1^{(\rm{ncoh})}$, $\bar{A}_2^{(\rm{ncoh})}$ and $\bar{A}_3^{(\rm{ncoh})}$ are defined as \eqref{eq:X},  \eqref{eq:Y} and \eqref{eq:Z} in Appendix \ref{app:g}, respectively.}
	\end{proposition}
	\begin{proof}
		\textcolor{black}{See Appendix \ref{app:g}.}
	\end{proof}
	\textcolor{black}{When dechirp is performed on the entire LoRa symbol, the SER can be obtained by replacing $\sigma_{M-3}^{(\rm{ncoh})}$ by $\sigma_{M-1}^{(\rm{ncoh})}$ in $\bar{A}_1^{(\rm{ncoh})}$ and setting $\bar{A}_2^{(\rm{ncoh})}=0$, $\bar{A}_3^{(\rm{ncoh})}=0$.} \textcolor{black}{This case achieves the lowest SER. However, additional pilot symbols are required for FAS channel estimation.} 
	\subsection{\textcolor{black}{SER Analysis of Coherent Receiver}}
	\textcolor{black}{This subsection focuses on the performance analysis of the proposed system under coherent detection.  Similar to the non-coherent case, the SER in \eqref{eq:ser_coh} can be rewritten as}
	\begin{equation}\label{eq:ser-exact-coh}\color{black}
		\begin{split}
			&{P}_{\rm{s}}^{(\rm{coh})}\\
			&\quad={\rm{E}}_{|\hat{h}|}\left\{1\!-\!{\rm{Pr}}\left(\max\limits_{k\in\mathcal{K}, k\neq{m}}\!X^{(\rm{coh})}_{k|\,|\hat{h}|=r}\leq{X}^{(\rm{coh})}_{m|\,|\hat{h}|=r}\right)\right\}\!\!\\
			&\quad={\rm{E}}_{|\hat{h}|}\left\{{\rm{E}}_{X_{m|\,|\hat{h}|=r}^{(\rm{coh})}}\left\{1-F_{\hat{X}^{(\rm{coh})}_{|\,|\hat{h}|=r}}\left(x^{(\rm{coh})}_{m|\,|\hat{h}|=r}\right)\right\}\right\}\\
			&\quad\triangleq{\rm{E}}_{|\hat{h}|}\left\{P_{\text{s}|\,|\hat{h}|=r}^{(\rm{coh})}\right\}, \forall m\in\mathcal{M},
		\end{split}
	\end{equation}
	\textcolor{black}{where $\hat{X}_{|\,|\hat{h}|=r}^{(\rm{coh})}=\max\limits_{k\in\mathcal{K},k\neq{m}}X_{k|\,|\hat{h}|=r}^{(\rm{coh})}$. $F_{\hat{X}_{|\,|\hat{h}|=r}^{(\rm{coh})}}(\cdot)$ denotes the conditional CDF of $\hat{X}_{|\,|\hat{h}|=r}^{(\rm{coh})}$. Different from the non-coherent case, both ${X}_{k|\,|\hat{h}|=r}^{(\rm{coh})}$ and ${X}_{m|\,|\hat{h}|=r}^{(\rm{coh})}$ are Gaussian distributed random variables with $X_{k|\,|\hat{h}|=r}^{(\text{coh})}\sim\mathcal{N}\left(\mathrm{Re}\left(\xi_{m-k}\right)r,\frac{N_0}{2}\right)$ and $X_{m|\,|\hat{h}|=r}^{(\text{coh})}\sim\mathcal{N}\left(Dr,\frac{N_0}{2}\right)$. Then similar to the proof of Lemma \ref{lemma:con-ser-approx}, the conditional CDF of $\hat{X}_{|\,|\hat{h}|=r}^{(\rm{coh})}$ can be approximated by}
	\begin{multline}\label{cdf-Xr-coh}\color{black}
			F_{\hat{X}_{|\,|\hat{h}|=r}^{(\rm{coh})}}\!(z)
			\approx\left[1-2Q\left(\frac{\mathrm{Re}\left(\xi_{\pm1}\right)r-z}{\sqrt{\frac{N_0}{2}}}\right)\right]\\ 
			\times H\left(z-\sigma_{M-3}^{(\rm{coh})}\,\right)\!,
	\end{multline}
	\textcolor{black}{where the maximum of bins with indices $k\in\mathcal{K}$, $k\neq m$ and $k\neq m\pm1$  is also approximated by a Heaviside step function. In this case, $\sigma_{M-3}^{(\rm{coh})}$ can be computed as the largest value of $M-3$ Gaussian random variables, each distributed as $\mathcal{N}\left(0,\frac{N_0}{2}\right)$. This value can be calculated by applying the extreme value theorem\footnote{\textcolor{black}{The asymptotic distribution of i.i.d. random variables $\big\{X_1,\cdots,X_{\eta}\big\}$ will only be one of three extreme value distributions (i.e., Frechet, Gumbel or Weibull). Moreover, it can be proved that the Gaussian distribution belongs to the maximum domain of attraction of Gumbel distribution, i.e., $X_n\in\mathcal{N}(0,\sigma)$, $\mathrm{Pr}\left(\frac{\max\{X_1,\cdots,X_{\eta}\}-\beta_{\eta}}{\alpha_{\eta}}\right){\to}\mathcal{G}(0,1)$ with $F_{\mathcal{G}}(x)=e^{-e^{-x}}$, for the normalized factor $\alpha_n=\sigma\Phi^{-1}\left(1-\frac{1}{n}\right)$ and $\beta_{\eta}=\sigma\Phi^{-1}\left(1-\frac{1}{e{\eta}}\right)-\sigma\Phi^{-1}\left(1-\frac{1}{\eta}\right)$, when $n\to\infty$. Thus, we have $\max\{X_1,\cdots,X_{\eta}\}\sim\mathcal{G}(\alpha_{\eta},\beta_{\eta})$. Finally, the mean of Gumbel distribution is $\alpha_{\eta}+\gamma\beta_{\eta}$, which is given in \eqref{eq:sigma-M3-coh}.}}\cite[Section 10.5]{bibitem52}, which is given by}
	\begin{equation}\label{eq:sigma-M3-coh}\color{black}
		\begin{split}
			{\sigma}_{M-3}^{(\rm{coh})}=\sqrt{\frac{N_0}{2}}&\left[(1-\gamma)\Phi^{-1}\left(1-\frac{1}{M-3}\right)\right.\\
			&\qquad\quad\left.+\gamma\Phi^{-1}\left(1-\frac{1}{e(M-3)}\right)\right],
		\end{split}
	\end{equation}
	\textcolor{black}{where $\Phi^{-1}(\cdot)$ is the inverse CDF of the Gaussian distribution with unit variance. Then the conditional SER can be obtained by integrating \eqref{cdf-Xr-coh} over the PDF of $X_{m|\,|\hat{h}|=r}^{(\text{coh})}$, resulting in}
	\begin{multline}\label{eq:con-ser-coh}\color{black}
		\!\!\!\!{P}_{{\rm{s}}|\,|\hat{h}|=r}^{(\rm{coh})}\approx1- \frac{1}{\sqrt{\pi N_0}}\int^{\infty}_{\sigma_{M-3}^{(\rm{coh})}}\left[1-2Q\left(\frac{\mathrm{Re}\left(|\xi_{\pm1}|\right)r-z}{\sqrt{\frac{N_0}{2}}}\right)\right]\\
		\color{black}\times e^{-\frac{\left(z-Dr\right)^2}{N_0}}dz.
	\end{multline}
	\textcolor{black}{The closed-form expression of the conditional SER in \eqref{eq:con-ser-coh} is given by the following lemma.}
	\begin{lemma}\label{lemma:ncon-ser-approx2}\color{black}
		The closed-form expression of the conditional SER  $P_{{\rm{s}}||\hat{h}|=r}^{(\rm{coh})}$ is given by \eqref{eq:con-ser-approx2-coh} at the top of the next page,
		\begin{figure*}[!h]
			\begin{equation}\label{eq:con-ser-approx2-coh}\color{black}
				{P}_{{\rm{s}}|\,|\hat{h}|=r}^{(\rm{coh})}\approx1\!-{\frac{1}{2}\,\underbrace{\mathrm{erfc}\left(\frac{{\sigma_{M-3}^{(\rm{coh})}}-Dr}{\sqrt{N_0}}\right)}_{A_{1}^{(\rm{coh})}}}+\sum^{2}_{i=1}\frac{\tilde{\alpha}_i}{\sqrt{1+2\tilde{\beta}_i}}\underbrace{e^{-\frac{2\tilde{\beta}_i\left(Dr-\mathrm{Re}\left(|\xi_{\pm1}|\right)r\right)^2}{N_0(1+2\tilde{\beta}_i)}}\mathrm{erfc}\left( \frac{{\sigma_{M-3}^{(\rm{coh})}}(1+2\tilde{\beta}_i)-Dr-2\tilde{\beta}_i\mathrm{Re}\left(|\xi_{\pm1}|\right)r}{\sqrt{N_0(1+2\tilde{\beta}_i)}} \right)}_{A_{i+1}^{(\rm{coh})}}
			\end{equation}
			\hrule
		\end{figure*}
		where $\tilde{\alpha}=\left[\frac{1}{12},\frac{1}{4}\right]$, $\tilde{\beta}=\left[\frac{1}{2},\frac{2}{3}\right]$.
	\end{lemma}
	\begin{proof}\color{black}
		This is done by utilizing the exponential-type approximation of $Q(\cdot)$ as $Q(x)\approx\sum^{2}_{i=1}\tilde{\alpha}_{i}e^{-\tilde{\beta}_{i}x^2}$. Moreover, the integral is solved by \eqref{int-pow-exp}.
	\end{proof}
	
	\textcolor{black}{By averaging \eqref{eq:con-ser-approx2-coh} over \eqref{eq:f-hstar}, the SER is obtained as follows.}
	\begin{proposition}\label{Ps-Cf}\color{black}
		The closed-form SER of the proposed LoRa-FAS under coherent detection can be approximated by
		\begin{multline}\label{ser-coh}\color{black}
			{P}^{(\rm{coh})}_{\rm{s}}\approx1-\sum^{B}_{b=1}\sum^{B}_{\tilde{b}=1}\,(-1)^{\tilde{b}-1}\binom{B-1}{\tilde{b}-1}\\
			\times\frac{1}{\mu^2}\left({\bar{A}_1^{(\rm{coh})}}-\sum^{2}_{i=1}\frac{2\alpha_i}{\sqrt{1+2\tilde{\beta}_i}}\bar{A}_{i+1}^{(\rm{coh})}\right),
		\end{multline}
		where $\bar{A}_{1}^{(\rm{coh})}$ is obtained by substituting $\sigma^{(\rm{ncoh})}_{M-3}$ in (\ref{eq:X}) by $\sigma^{(\rm{coh})}_{M-3}$. $\bar{A}_{i}^{(\rm{coh})}$, $i\in\{1,2\}$ is given on the top of next page. Moreover, we have $\tilde{\Xi}_{1,i}=\frac{\tilde{b}}{\mu^2}+\frac{2\tilde{\beta}_i\left(D-\mathrm{Re}\left(\xi_{\pm1}\right)\right)^2}{N_0(1+2\tilde{\beta}_i)}+\frac{\Theta_1\left(D+2\tilde{\beta}_i\mathrm{Re}\left(\xi_{\pm1}\right)\right)^2}{N_0(1+2\tilde{\beta}_i)}$, $\tilde{\Xi}_{2,i}=-\frac{2\tilde{b}}{\mu^2}\delta_{b}-\frac{\Theta_12\sigma^{(\rm{coh})}_{M-3}\left(D+2\tilde{\beta}_i\mathrm{Re}\left(\xi_{\pm1}\right)\right)}{N_0}-\frac{\Theta_2\left(D+2\tilde{\beta}_i\mathrm{Re}\left(\xi_{\pm1}\right)\right)}{\sqrt{N_0(1+2\tilde{\beta}_i)}}$ and $\tilde{\Xi}_{3,i}=\frac{\tilde{b}\left(1-\mu^2\right)}{2\mu^2}\delta_{b}+\frac{\Theta_1\left(\sigma^{(\rm{coh})}_{M-3}\right)^2\left(1+2\tilde{\beta}_i\right)}{N_0}+\frac{\Theta_2\sigma^{(\rm{coh})}_{M-3}\sqrt{1+2\tilde{\beta}_i}}{\sqrt{N_0}}+\Theta_3$, where $\Theta_1$, $\Theta_2$ and $\Theta_3$ have been given in Appendix \ref{app:g}.
	\end{proposition}
	\begin{proof}\color{black}
		This proof is similar to Appendix \ref{app:g}.
	\end{proof}
	\begin{figure*}[!ht]
		\begin{equation}\label{eq:B-coh}\color{black}
			\begin{split}
				\bar{A}_{i+1}^{(\rm{coh})}&=\!\int^{\infty}_{\hat{\delta}}\left(r-\delta_b\right)e^{-\frac{\tilde{b}}{\mu}\left(r-\delta_{b}\right)^2-\frac{2\tilde{\beta}_i}{N_0(1+2\tilde{\beta}\beta_i)}\left(Dr-\mathrm{Re}\left(\xi_{\pm1}\right)r\right)^2}\mathrm{erfc}\left(\frac{(1+2\tilde{\beta}_i){\sigma}^{(\rm{coh})}_{2^M-3}-\left(D+2\tilde{\beta}_i\mathrm{Re}\left(\xi_{\pm1}\right)\right)r}{\sqrt{N_0(1+2\tilde{\beta}_i)}}\right)dr\\
				&\approx\tilde{\Xi}^{-\frac{3}{2}}_{1,i}e^{\frac{\tilde{\Xi}^2_{2,i}}{4\tilde{\Xi}_{1,i}}-\tilde{\Xi}_{3,i}}\left(\sqrt{\tilde{\Xi}_{1,i}}e^{-\frac{\left(\frac{\tilde{\Xi}_{2,i}}{2}+\tilde{\Xi}_{1,i}\hat{\delta}\right)}{2\tilde{\Xi}_{1,i}}}-\sqrt{\pi}\left(\frac{\tilde{\Xi}_{2,i}}{2}+\tilde{\Xi}_{1,i}\delta_{b}\right)\mathrm{erfc}\left(\frac{\frac{\tilde{\Xi}_{2,i}}{2}+\tilde{\Xi}_{1,i}\hat{\delta}}{\sqrt{\tilde{\Xi}_{1,i}}}\right)\right)
			\end{split}
		\end{equation}
		\hrule
	\end{figure*}
	
	\section{Simulation Results}\label{sec:results}
	In this section, the performance of the proposed LoRa-FAS is assessed through Monte-Carlo simulations and theoretical results. The analytical results were based on \eqref{ser-ncoh} and \eqref{ser-coh}. All parameters are specified in the figure captions. In the results, lines correspond to Monte-Carlo results, while markers specify the results from our analytical results. \textcolor{black}{Note that the non-coherent detection scheme in the conventional LoRa system does not require CSI; therefore, the entire symbol duration is available for data transmission. For fairness, we assume that the coherent scheme in the conventional LoRa system adopts the same pilot insertion pattern as LoRa-FAS.}
	\begin{figure}[!t]
		\centering
		\includegraphics[width=0.8\linewidth]{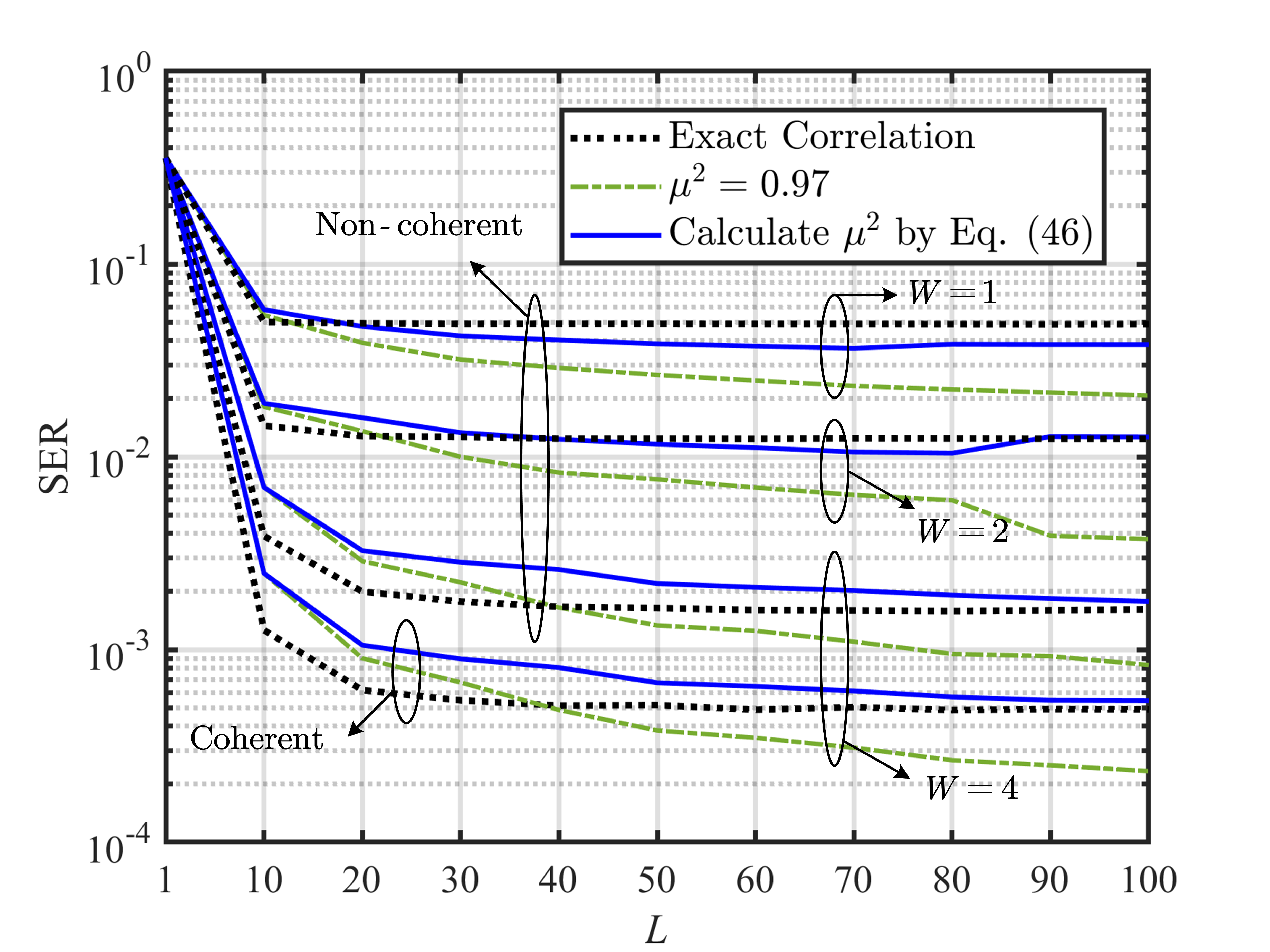}
		\caption{The SER for proposed LoRa-FAS with non-coherent and coherent symbol detection schemes under different number of ports $L$, where 
			$\Gamma=-10$ dB, $SF=7$, $W=\{1,2,4\}$ and $1-D=\frac{1}{16}$.}\label{fig:3}
	\end{figure}
	\begin{figure*}[!t]\label{Sim4}\color{black}
		\centering
		\subfloat[]{\includegraphics[width=0.3333\linewidth]{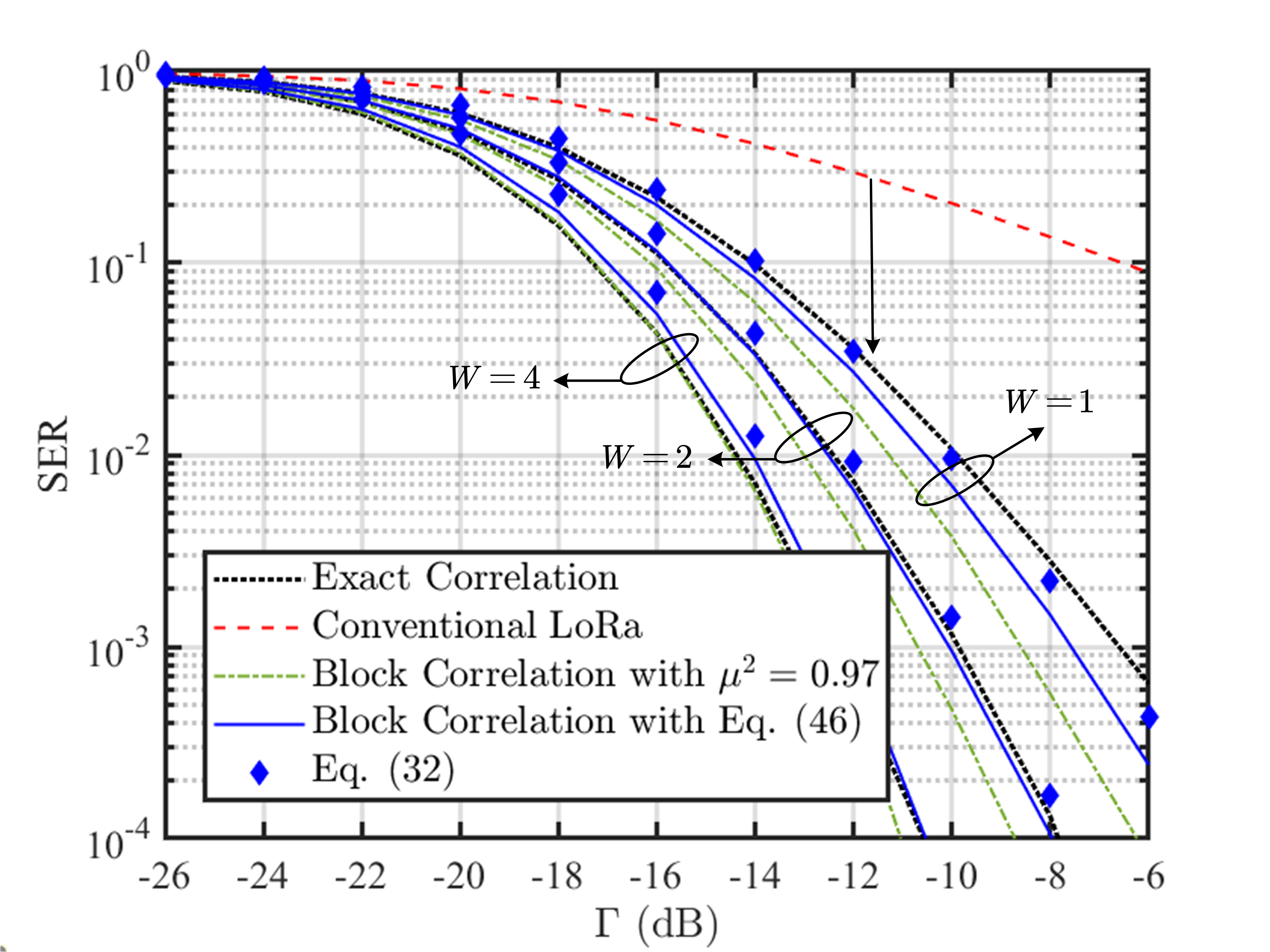}%
			\label{sim:4a}}
		\hfil
		\subfloat[]{\includegraphics[width=0.3333\linewidth]{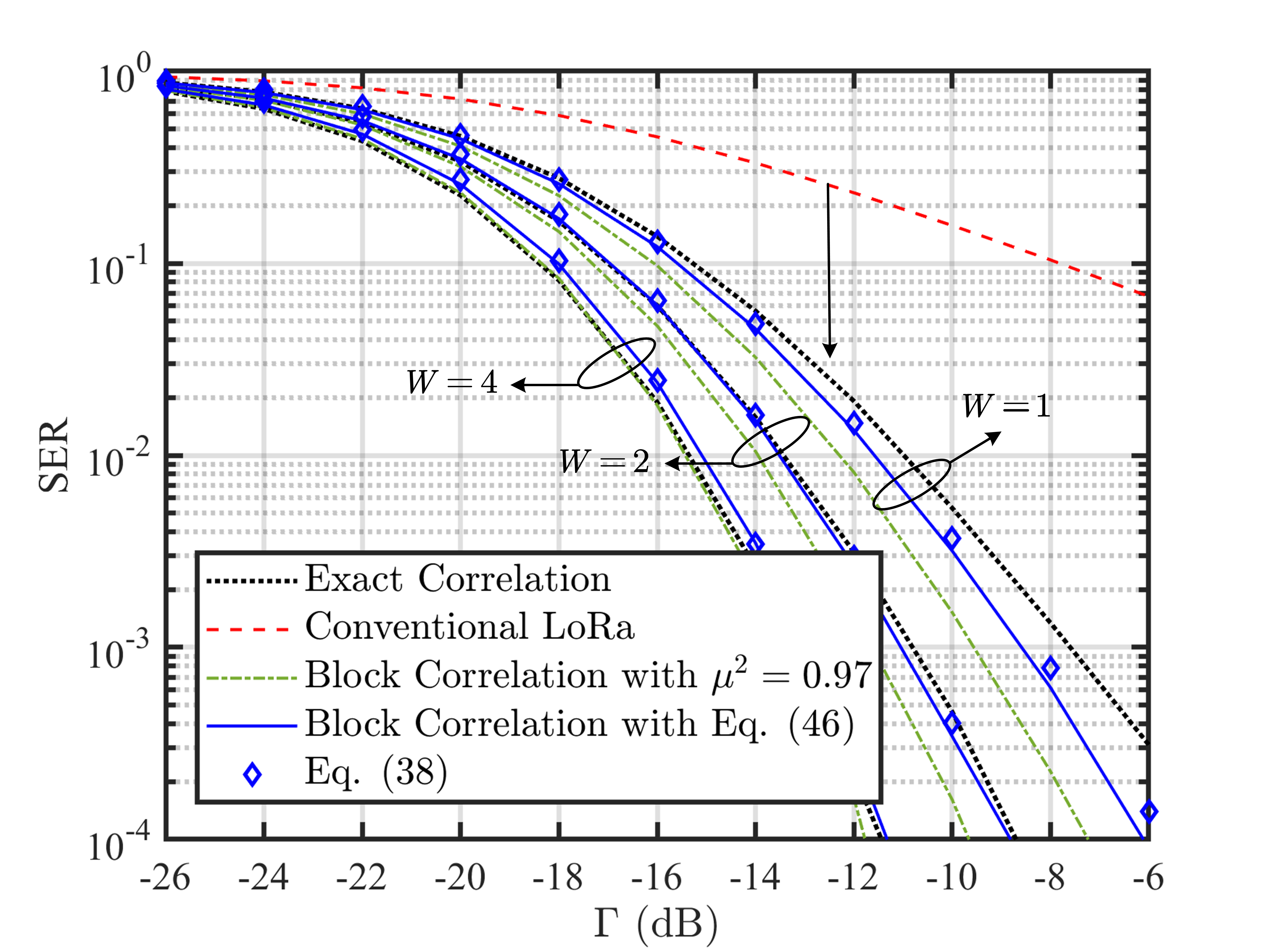}%
			\label{sim:4b}}
		\hfil
		\subfloat[]{\includegraphics[width=0.3333\linewidth]{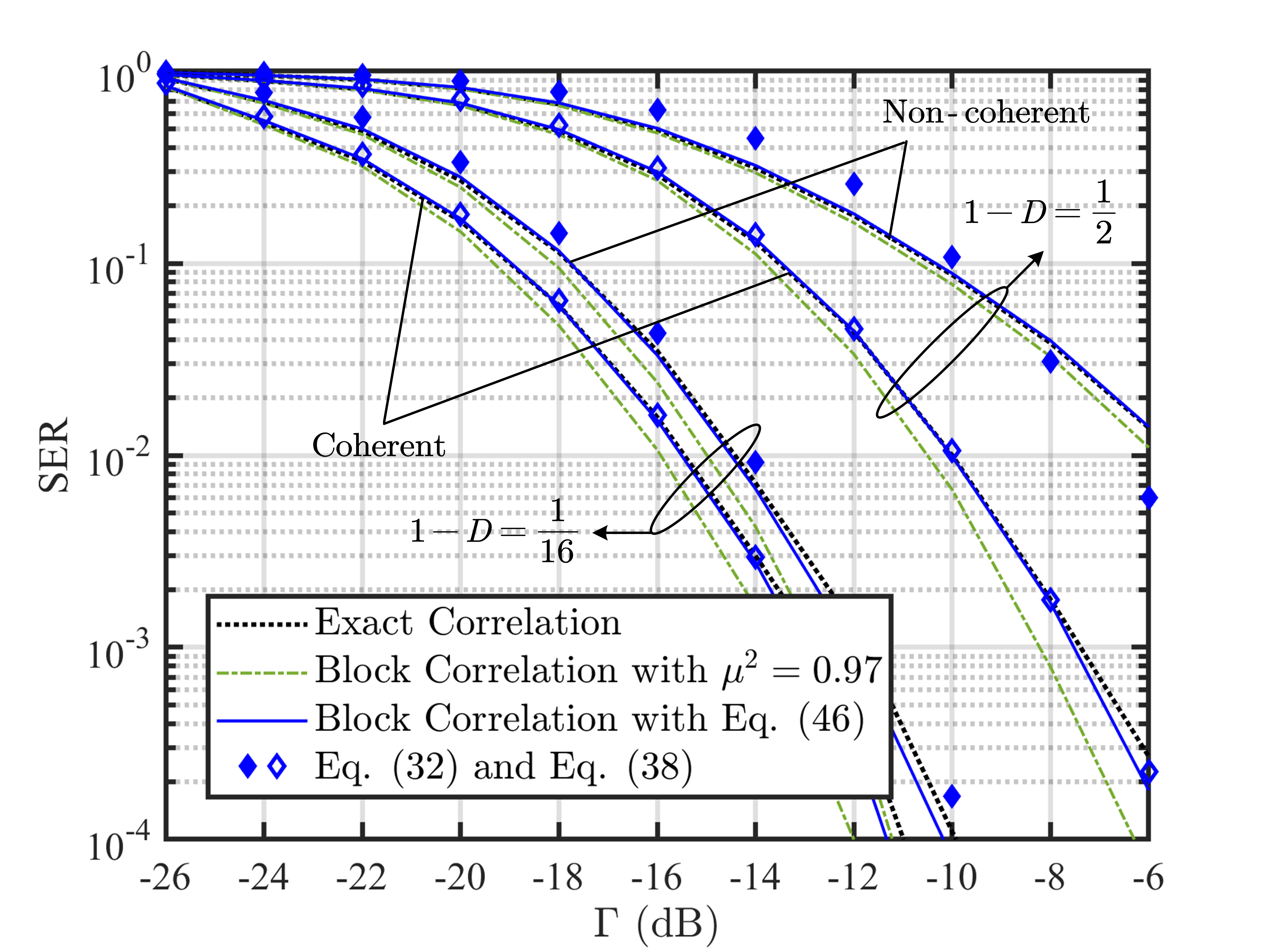}%
			\label{sim:4c}}
		\caption{The SER for the proposed LoRa-FAS versus different SNR. (a) Non-coherent detection scheme, with ${\rm SF}=8$, $L=50$, $W=\{1,2,4\}$ and $1-D=\frac{1}{16}$. (b) Coherent detection scheme, with ${\rm SF}=8$, $L=50$, $W=\{1,2,4\}$ and $1-D=\frac{1}{16}$. (c) Both non-coherent and coherent detection schemes, with ${\rm SF}=8$, $L=50$, $W=2$, and $1-D=\left\{\frac{1}{16},\frac{1}{2}\right\}$.}
		\label{fig:4}
	\end{figure*}
	\begin{figure*}[!t]\label{fig:sim5}\color{black}
		\centering
		\subfloat[]{\includegraphics[width=0.3333\linewidth]{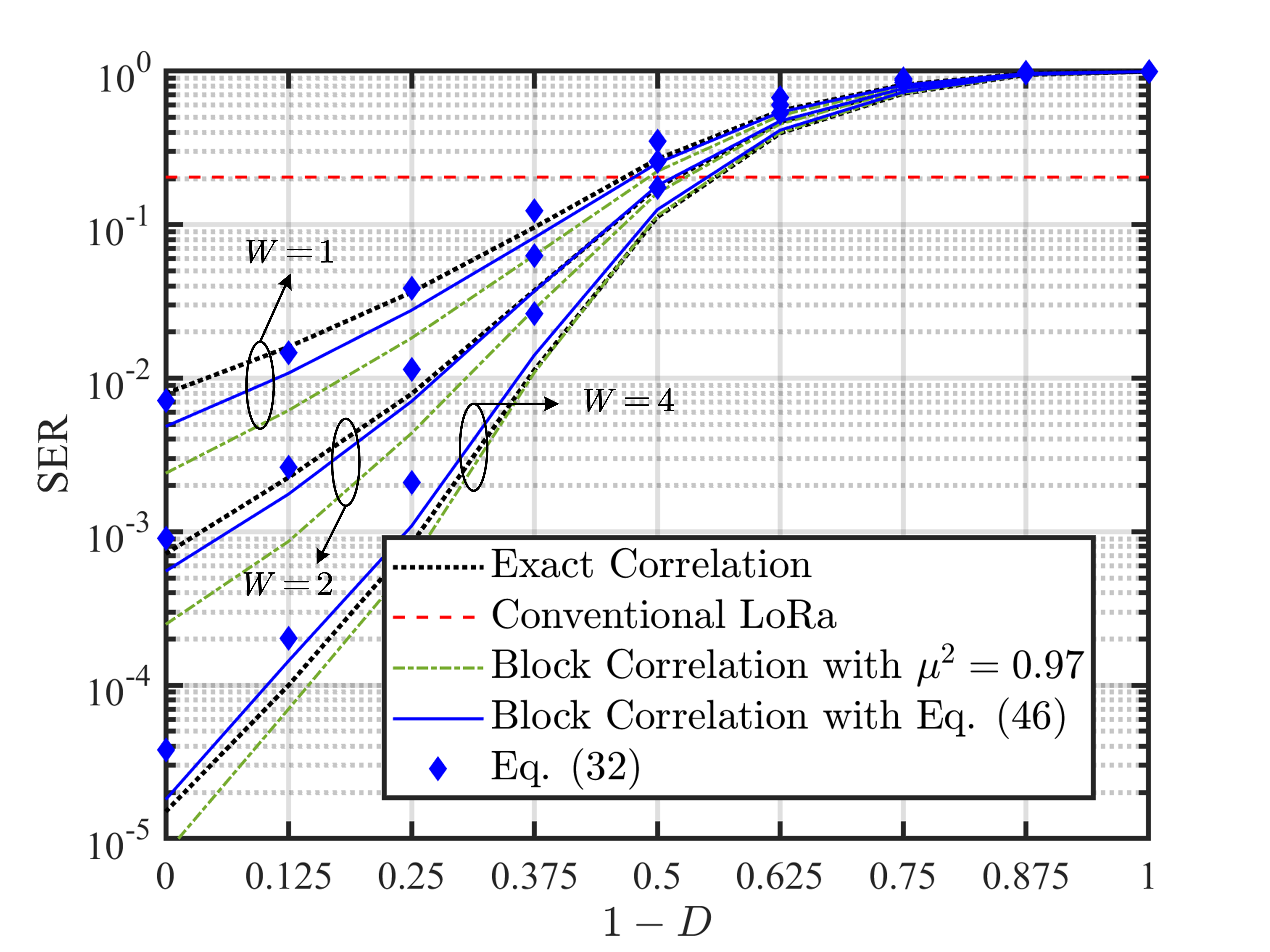}%
			\label{sim:5a}}
		\hfil
		\subfloat[]{\includegraphics[width=0.3333\linewidth]{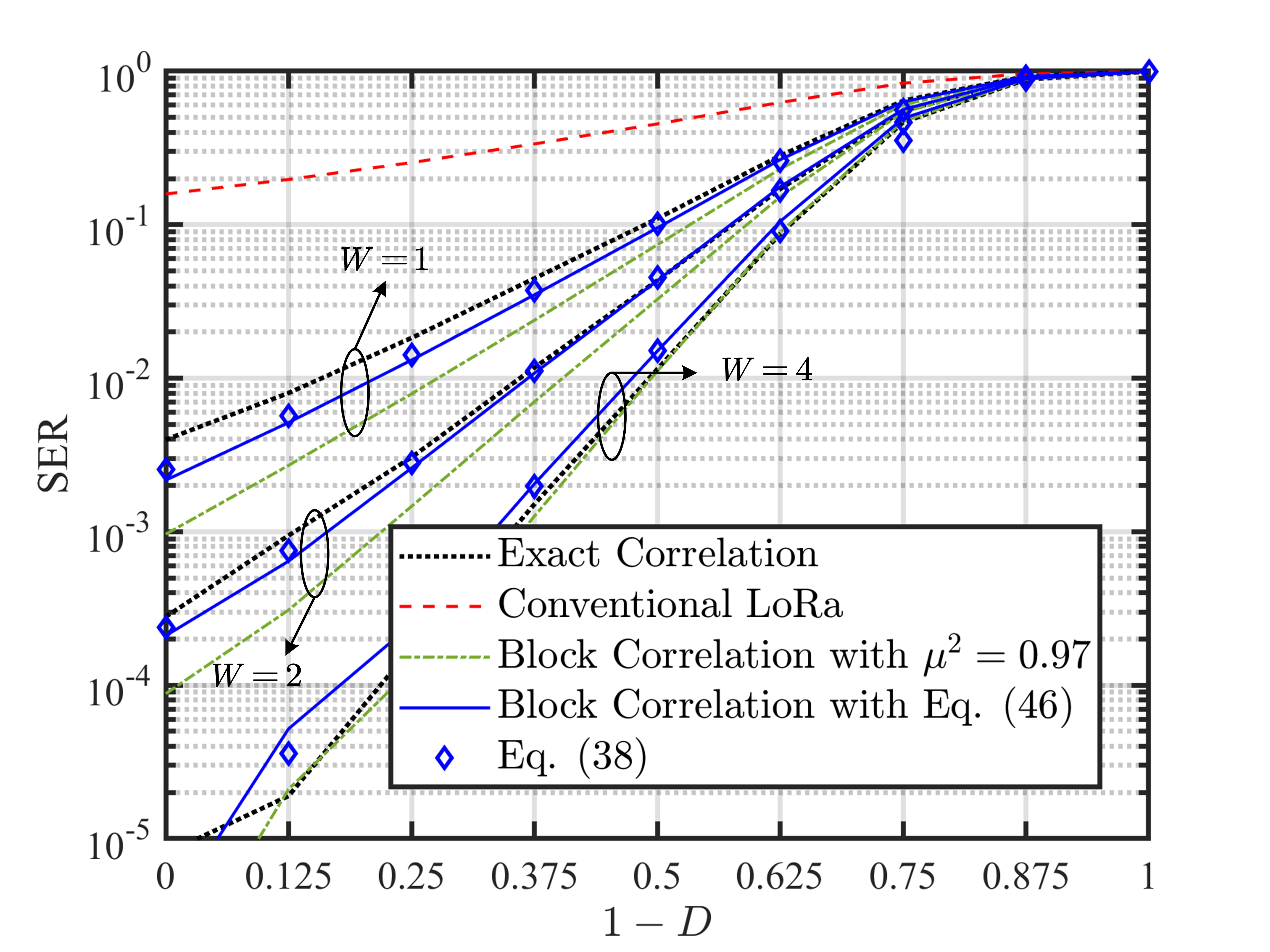}%
			\label{sim:5b}}
		\caption{The SER for the proposed LoRa-FAS versus different SNR pilot occupy occupation. (a) Non-coherent detection scheme, with $\Gamma=-10$ dB, $SF=8$, $L=50$ and $W=\{1,2,4\}$. (b) Coherent detection scheme, with $\Gamma=-10$ dB, $SF=8$, $L=50$ and $W=\{1,2,4\}$. In (a) and (b), to investigate the SER under different proportion of pilot overhead from $0$ to $1$, we neglect the relationship between $P$, $M$, $U$ and $D$ in \eqref{w-m} and directly set different value of $1-D$.}
		\label{fig:5}
	\end{figure*}
	
	\textcolor{black}{Fig.~\ref{fig:3} illustrates the SER performance of the proposed LoRa-FAS under varying numbers of ports, assuming perfect CSI for FAS port selection and LoRa coherent detection.} Consistent with the results presented in Fig.~2(b) for the case where $W=6$, the block correlation model with a fixed $\mu^2$ approaches the exact model as $L$ increases. However, further increasing $L$ beyond certain thresholds results in a smaller SER that appears to deviate from the exact model. \textcolor{black}{As $L$ increases, the SER decreases continuously without converging to a finite lower bound, which contradicts the intuition that spatial gain within a limited length is inherently constrained.} A similar trend is observed when $\mu^2$ is calculated using \eqref{app:new_mu}, but in this case, the SER does not continue to decrease as $L\to\infty$. \textcolor{black}{Furthermore, it is evident that as $L$ increases, the proposed choice of $\mu^2$ brings the results closer to the exact correlation case generated by Clarke's model. These findings indicate the validity of the proposed method for selecting $\mu^2$. Moreover, the exact correlation model exhibits faster coverage with increasing $L$, while the block-correlation models demonstrate relatively slower coverage, which may result in decreased accuracy for a small number of ports. Nevertheless, these models can still be utilized to understand the performance limitations of the FAS. Additionally, Monte-Carlo results for non-coherent and coherent detection schemes are provided for $W=4$. As anticipated, the coherent scheme outperforms the non-coherent case by leveraging additional information for symbol detection.}
	
	\begin{table*}[t]\color{black}
		\footnotesize
		\caption{The SER for the proposed LoRa-FAS versus different length of pilot sequence via Monte-Carlo simulations}\vspace{-1em}
		\begin{center}
			\begin{tabular}{m{1.8cm}<{\centering}|m{1.65cm}<{\centering}|m{1.65cm}<{\centering}|m{1.65cm}<{\centering}|m{1.65cm}<{\centering}|m{3cm}<{\centering}|m{1.65cm}<{\centering}|m{1.65cm}<{\centering}}
				\Xhline{0.8px}
				{Schemes} &\multicolumn{2}{c|}{{LoRa-FAS (Non-coherent)}} &\multicolumn{2}{c|}{{LoRa-FAS (Coherent)}} & {LoRa (Non-coherent)} & \multicolumn{2}{c}{{LoRa (Coherent)}} \\
				\hline
				\diagbox{}{CSI} & {Perfect} & {Least squares} & {Perfect} & {Least squares} & \textbf{/} & {Perfect} & {Least squares}\\
				\hline
				$P=2^6$	& $6.6\times10^{-4}$ & $2.7\times10^{-3}$  & $3.1\times10^{-4}$ & $2.0\times10^{-3}$ & \multirow{3}{*}{$8.9\times10^{-2}$}  & $7.6\times10^{-2}$  & $1.0\times10^{-1}$ \\
				\cline{1-5}\cline{7-8}
				$P=2^7$	& $1.0\times10^{-3}$ & $2.5\times10^{-3}$  & $4.6\times10^{-4}$ & $1.6\times10^{-3}$ & &  $8.6\times10^{-2}$  & $9.8\times10^{-2}$ \\
				\cline{1-5}\cline{7-8}
				$P=2^8$	& $3.0\times10^{-3}$ & $4.6\times10^{-3}$  & $1.2\times10^{-3}$ & $2.3\times10^{-3}$ & &  $1.1\times10^{-1}$  & $1.2\times10^{-1}$ \\
				\Xhline{0.8px}
			\end{tabular}
		\end{center}\label{tab0}\vspace{0em}
		\begin{tablenotes}
			\footnotesize
			\item The parameters are set as $\Gamma=-6$ dB, $SF=8$, $L=50$, $W=1$. Setting $U=4$, therefore $P=\left\{2^6,2^7,2^8\right\}$ corresponds to $1-D=\left\{\frac{1}{16},\frac{1}{8},\frac{1}{4}\right\}$, respectively.
		\end{tablenotes}
		\label{tab:2}
	\end{table*}
	
	\textcolor{black}{In Fig.~\ref{fig:4}, the SER performance of the proposed LoRa-FAS is evaluated under different SNRs. First, the analytical results for both non-coherent and coherent schemes closely match the Monte-Carlo simulations in most cases. However, the accuracy of the analytical results decreases for relatively small values of $W$ and in the high SNR region. The former is due to the decreased fitting performance of the block-correlation model, as discussed in Fig.~\ref{fig:3}, while the latter arises from approximations made during the derivation of the closed-form expressions.} Second, we compare the conventional LoRa system with a single fixed antenna and the proposed LoRa-FAS. The results demonstrate that even a small FAS length, such as $W=1$, can yield significant performance improvements. This is because the correlation detector accumulates the diversity gain of the LoRa symbol at each sample point. Note that the number of sample points $M$ is typically large, making this gain quite substantial. \textcolor{black}{Third, comparing the non-coherent and coherent schemes at the same $W$ (e.g., $W=4$) and target SER (e.g., SER$=10^{-4}$) in Fig.~4(a) and Fig.~4(b), it can be observed that the performance gap between two schemes is limited to approximately $1$ dB. However, in Fig.~4(c), when the pilot occupies half a symbol duration, the advantage of the coherent scheme becomes more apparent. The reasons have been explained in Section \ref{sec:model-receiver}. Therefore, when the channel coherence time is large, we can distribute the pilots over more symbols to achieve better SER performance. In this case, the non-coherent and coherent schemes can be considered interchangeable. However, when the channel coherence time is short, we suggest using the coherent scheme to reduce the channel estimation time for each port.}
	
	In Fig.~\ref{fig:5}, the SER performance of the proposed LoRa-FAS is evaluated under varying pilot sequence overheads within a LoRa symbol. \textcolor{black}{Based on the simulation results, several observations can be made.} First, as the pilot overhead $1-D$ increases, the average SER for non-coherent and coherent detection schemes degrades. This degradation is attributed to the fact that pilot sequences do not carry data information, thus reducing the number of samples available for symbol detection. \textcolor{black}{Second, although determining the optimal pilot length remains challenging, the benefits of embedding pilots within the symbol are evident. For example, four symbols are required to observe one port when $1-D=0.25$ and $P=M$ within the data portion $D$. Recall that the throughout is defined as $R=\frac{\log_2(M)}{T_\mathrm{s}}(1-P_\mathrm{s})$. Then the conventional continuous pilot insertion reduces the effective throughput to $\frac{3}{4}R$ due to the pilot overhead. In contrast, the embedded pilot scheme preserves the full throughput $R$ since $P_\mathrm{s}$ is typically small.}
	
	\textcolor{black}{Table \ref{tab:2} evaluates the SER performance for different pilot sequence lengths, which affect channel estimation accuracy and pilot overhead simultaneously. The objective is to examine the existence of a trade-off between these two factors. From the numerical results, several key observations can be drawn. First, for all cases employing least squares channel estimation, a pilot length of $P=2^7$ yields the lowest SER, indicating the presence of a trade-off. Specifically, longer pilot sequences improve CSI accuracy for both FAS port selection and LoRa coherent detection, but they also increase pilot overhead, which may result in a higher SER. Second, given that the SF for both data and pilot symbols is generally greater than or equal to $7$, the results show that the SER for $P=8$ exceeds that of $P=7$. In this regime, the enhanced channel estimation accuracy cannot compensate for the negative impact of increased pilot overhead. Therefore, we recommend minimal pilot sequence length when the total observation time across all ports remains within the channel coherence time. Third, although practical LoRa systems must balance pilot overhead and channel coherence time constraints, the proposed LoRa-FAS scheme remains highly competitive due to its substantial performance gain and the convenience of preserving the existing LoRa PHY frame structure.}
	
	\section{Conclusion}\label{sec:conclude}
	\textcolor{black}{In this paper, we propose a LoRa-FAS system motivated by the capability of FAS to provide additional spatial gains without the complexity of integrating MIMO. Approximate closed-form expressions for the CDF and PDF of the Rx-SISO-FAS channel are derived in a tractable form. Furthermore, closed-form approximations of the SER for the proposed LoRa-FAS under both coherent and non-coherent detection schemes are obtained. Numerical results demonstrate that LoRa-FAS significantly outperforms the conventional LoRa system, even when the size of the fluid antenna is relatively small. Additionally, we investigate the impact of embedded pilot symbol overhead and placement on SER performance. We recommend employing coherent detection for scenarios with relatively short channel coherence time, which requires a larger proportion of embedded pilot symbols. In contrast, spreading the pilot sequence over more symbols is advisable for a longer channel coherence time. Moreover, coherent LoRa-FAS requires accurate CSI, whereas the non-coherent scheme can operate similarly to conventional LoRa systems without CSI. Then any improvements in channel estimation accuracy can further enhance the performance of the LoRa-FAS system. Finally, since the analytical results are derived under the assumption of perfect CSI, further investigation into the effects of channel estimation errors and advanced channel estimation algorithms is essential for practical implementation and warrants future work.}
	
	\appendices
	\section{Selection of $\mu^2$}\label{app:a}
	As discussed in Section \ref{sec:model}, the exact correlation matrix $\mathbf{\Sigma}\in\mathbb{R}^{L\times L}$ is derived using Clarke's model, where the correlation factor between the $i$-th and $j$-th ports is modeled as
	\begin{equation}\label{eq:sigma-sinc}
		(\mathbf{\Sigma})_{i,j}\approx \mathrm{sinc}\left(\frac{2\pi(i-j)W}{L-1}\right).
	\end{equation}
	Let $\{\rho_l\}_{l=1}^{L}$ and $\{\rho_b\}_{b=1}^{B}$ denote eigenvalues of $\mathbf{\Sigma}$ and the largest $B$ of them, respectively. Both of them are arranged in descending order. $B$ represent the number of eigenvalues that exceed a given threshold $\rho_\text{th}$. Then $\mathbf{\Sigma}$ can be effectively approximated by fitting these eigenvalues, resulting in a structure of \cite{bibitem42}
	\begin{equation}\label{out-corr}\color{black}
		\begin{split}
			\tilde{\mathbf{\Sigma}}=
			\left(
			\!\!\begin{array}{cccc}
				\mathbf{A}_1 & \mathbf{0} & \cdots & \mathbf{0}\\
				\mathbf{0} & \mathbf{A}_2 & \cdots & \mathbf{0}\\
				\vdots & \vdots & \ddots & \vdots\\
				\mathbf{0} & \mathbf{0} & \cdots & \mathbf{A}_{B}
			\end{array}
			\!\!\right)\in\mathbb{R}^{B\times B},
		\end{split}
	\end{equation}
	\begin{equation}\label{in-corr}\color{black}
		\begin{split}
			{\mathbf{A}}_{b}=
			\left(
			\!\!\begin{array}{cccc}
				1 & \mu^2 & \cdots & \mu^2\\
				\mu^2 & 1 & \cdots & \mu^2\\
				\vdots & \vdots & \ddots & \vdots\\
				\mu^2 & \mu^2 & \cdots & 1
			\end{array}
			\!\!\right)\in\mathbb{R}^{L_b\times L_b}.
		\end{split}
	\end{equation}
	\textcolor{black}{Similarly, let the eigenvalues of $\tilde{\mathbf{\Sigma}}$ be $\{\tilde{\rho_l}\}_{l=1}^{L}$, with a one-to-one correspondence to $\tilde{\rho}_{(b,l_b)}$, where $b=\{1,2,\cdots,B\}$, $l_b=\{1,2,\cdots,L_b\}$. $\{\tilde{\rho}_{(b,l_b)}\}_{l_b=1}^{L_b}$ are eigenvalues of a single block $\mathbf{A}_b$, which can be expressed as}
	\begin{equation}\label{eig-blk}
		\tilde{\rho}_{(b,l_b)}=\!\left\{\begin{aligned}
			&\left(L_b-1\right)\mu^2+1,\!&\!&l_b=1,\\
			&1-\mu^2,\!&\!&l_b=2,3,\dots,L_b.
		\end{aligned}\right.
	\end{equation}
	\textcolor{black}{The aim of \cite[Algorithm I]{bibitem33} is to obtain the results of}
	\begin{equation}\color{black}
		\mathop{\arg\min}_{L_1,L_2,\dots,L_B}\sum_{b=1}^{B}\left\|\rho_b-\tilde{\rho}_{(b,1)}\right\|^2_2,
	\end{equation}
	This approximation disregards the influence of smaller eigenvalues, which might become significant when $L$ is large. Specifically, let $\{\rho_b\}_{b=1}^B$ correspond to $\{\rho_{(b,1)}\}_{b=1}^B$, others $\rho_l$ for $l\neq\{1,\cdots B\}$ and ${\rho_{(b,l_b)}}$ for $l_b\neq1$ can be arranged randomly. Then the squared error of eigenvalues of $\mathbf{\Sigma}$ and $\tilde{\mathbf{\Sigma}}$ can be expressed as
	\begin{equation}\label{app:old_mu}
		\begin{split}
			\varepsilon=\sum_{l=1}^{L}&\left\|\rho_l-\tilde{\rho}_l\right\|^2_2=\sum_{b=1}^{B}\sum_{l_b=1}^{L_b}\left\|\rho_{(b,l_b)}-\tilde{\rho}_{(b,l_b)}\right\|^2_2\\
			&\overset{(a)}{\approx}\sum_{b=1}^{B}\sum_{l_b=2}^{L_b}\left\|\rho_{(b,l_b)}-\tilde{\rho}_{(b,l_b)}\right\|^2_2\\
			&\overset{(b)}{=}\sum_{b=1}^{B}\sum_{l_b=2}^{L_b}\left\|\rho_{(b,l_b)}-\left(1-\mu^2\right)\right\|^2_2\\
			&\overset{(c)}{\approx}(L-B)\left(1-\mu^2\right)^2\overset{(d)}{\approx} L\left(1-\mu^2\right)^2.
		\end{split}
	\end{equation}
	\textcolor{black}{where $(a)$ is obtained by neglecting the fitting error of the largest $B$ eigenvalues. $(b)$ is based on \eqref{eig-blk}, and $(c)$ is attributed to the fact that most of the smaller eigenvalues are close to $0$. $(d)$ is due to the fact that $\mathbf{\Sigma}$ is dominate by a small number of large eigenvalues, i.e., $L\gg B$ \cite{bibitem42}. Furthermore, \eqref{app:old_mu} implies that, for a fixed value of  $\mu^2$, the squared error increases as the number of ports increases. Moreover, this error does not possess an upper bound.}
	
	\textcolor{black}{We suggest that, before utilizing \cite[Algorithm I]{bibitem42}, a dynamic and judicious selection of $\mu^2$ should be performed. Specifically, our objective is to determine $\mu^2$ by}
	\begin{equation}\label{app:new_mu}
		\begin{split}
			\mu^2&=\mathop{\arg\min}\limits_{\mu^2}\sum_{b=1}^{B}\sum_{l_b=2}^{L_b}\left\|\rho_{(b,l_b)}-\left(1-\mu^2\right)\right\|^2_2\\
			&\overset{(e)}{=}1-\sum_{b=1}^{B}\sum_{l_b=2}^{L_b}\frac{\rho_{(b,l_b)}}{L-B}\approx1-\sum_{b=1}^{B}\sum_{l_b=2}^{L_b}\frac{\rho_{(b,l_b)}}{L}\\
			&\overset{(f)}{=}\frac{\sum_{b=1}^{B}\rho_{(l,b)}}{\sum_{b=1}^{B}\sum_{l_b=1}^{L_b}\rho_{(l,b)}},
		\end{split}
	\end{equation}
	\textcolor{black}{where $(e)$ is obtained by using the first-order optimality condition and solving
		$\partial_{\mu^2}\left({\sum_{b=1}^{B}\sum_{l_b=2}^{L_b}\left\|\rho_{(b,l_b)}-\left(1-\mu^2\right)\right\|^2_2}\right)=0.$} \textcolor{black}{Besides, $(f)$ is due to the characteristics of $\bf{\Sigma}$, where  $L=\sum_{l=1}^{L}\rho_l=\sum_{b=1}^{B}\sum_{l_b=1}^{L_b}\rho_{(b,l_b)}$. The final result of \eqref{app:new_mu} is the proportion of the eigenvalues greater than $\rho_\text{th}$. To avoid iterations, we choose the threshold as $\rho_\text{th}=1$.}
	\section{CDF of $\big|\hat{h}_b\big|$}\label{app:b}
	Our proof starts from \eqref{eq:h-channel}. let $z_b=\sqrt{x_{b,l_0}^2+y_{b,l_0}^2}$, then we have $z_b\sim \mathrm{Ra}\left(\frac{\sqrt{2}}{2}\right)$. Given $z_b$, ${|h_{b,l_b}|}_{|z_b}$ is a Rician distributed random variable, i.e., ${|h_{b,l_b}|}_{|z_b}\sim\mathrm{Ri}\Big(\mu\sqrt{z_b},\sqrt{\frac{1-\mu^2}{2}}\Big)$. Note that $\left\{{|h_{b,l_b}|}_{|z_b}\right\}_{l_b=1}^{L_b}$ are mutually independent. Therefore, the conditional CDF of ${|\hat{h}_{b}|}_{|z_b}=\max\limits_{l_b\in\mathcal{L}_b}{{|h_{b,l_b}|}_{|z_b}}$ can be expressed as $F_{{|h_{b,l_b}|}}(r|z_b)$ $=\prod_{l_b=1}^{L_b}F_{{|h_{b,l_b}|}}(r|z_b)$. Based on the law of total probability, the CDF of $|\hat{h}_b|$ is given by
	\begin{equation}\label{eq:F-lemma}
		\begin{split}
			F_{|\hat{h}_b|}\left(r\right)&=\int^{\infty}_{0}F_{{|h_{b,l_b}|}}(r|z_b)f_{|z_b|}(z_b)dz_b\\
			&=\int^{\infty}_{0}\!\!{\underbrace{\big(1-Q_1\left(a{z_b},br\right)\big)^{L_b}}_{g_0(z_b)=g_2(g_1(z_b))=(g_1(z_b))^{L_b}}}\!\!2z_be^{-z_b^2}dz_b,
		\end{split}
	\end{equation}
	where $a=\sqrt{\frac{2\mu^2}{1-\mu^2}}$ and $b=\sqrt{\frac{2}{1-\mu^2}}$. For large $L_b$, the term $g_0(z_b)$ approaches to a Heaviside step function due to the sigmoid behavior of the Marcum Q-function. \textcolor{black}{Then we will demonstrate the existence of an inflection point in $g_0(z_b)$ and select it as the step point.}
	
	Here, we first investigate the convexity of this term. Applying the chain rule for derivatives of composite functions, the second derivative of $g_0(z_b)$ can be represented as 
	\begin{multline}\label{2ed-dff-1-Q}
		\frac{\partial^2 g_0(z_b)}{\partial {z_b^2}}=L_b(g_1(z_b))^{L_b-2}\\
		\times\left((L_b-1)\left(\frac{\partial g_1(z_b)}{\partial {z_b}}\right)^2+g_1(z_b)\frac{\partial^2 g_1(z_b)}{\partial {z_b^2}}\right). 
	\end{multline}
	
	\textcolor{black}{When $az_b>br$, based on the convexity properties of Marcum Q-function in \cite[(17)]{app:b-1}, it follows that $\frac{\partial^2 g_1(z_b)}{\partial {z_b^2}}>0$ for $az_b>br$. Moreover, since $Q_1(az_b,br)\in[0,1)$, we have $g_1(z_b)=1-Q_1(az_b,br)>0$. Then the second derivative of $g_0(z_b)$ satisfies $\frac{\partial^2 g_0(z_b)}{\partial {z_b^2}}>0$, which implies that $g_0(z_b)$ is convex for $az_b>br$.} 
	
	\textcolor{black}{When $az_b<br$, we directly select the point at $z_b=0$. Using the partial differentials of Marcum Q-function \cite[(13)]{app:b-1}, we have 
	\begin{equation}
		\left\{\begin{aligned}
			&\frac{\partial g_1(z_b)}{\partial z_b}=-a^2z_b\left(Q_2(az_b,br)-Q_1(az_b,br)\right)\\
			&\frac{\partial^2 g_1(z_b)}{\partial z_b^2}=-a^2\left(Q_2(az_b,br)-Q_1(az_b,br)\right)-a^4z_b^2\\
			&\quad\quad\times(Q_3(az_b,br)-2Q_2(az_b,br)+Q_1(az_b,br)).
		\end{aligned}\right.
	\end{equation}	
	Substitute them into \eqref{2ed-dff-1-Q} and let $z_b=0$, we have
		\begin{equation}\label{2ed-dff-1-Q-0}
			\begin{split}
				\frac{\partial^2 g_0(0)}{\partial {z_b^2}}=a^2L_b&(1-Q_1(0,br))^{L_b-1}\\
				&\times\left(Q_1(0,br)-Q_2(0,br)\right).
			\end{split}
		\end{equation}
		Note that $Q_v(a,b)$ is monotonically increasing at $v$ for $a\geq 0$ and $b>0$. Therefore, we have $\frac{\partial^2 g_0(z_b)}{\partial {z_b^2}}<0$ at $z_b=0$ and $g_0(z_b)$ is concave at this point. Since $\frac{\partial^2 g_0(z_b)}{\partial {z_b^2}}$ is continuous, the inflection point is exist.}
	Then the inflection point is selected as the step point where the second derivative of $g_0(z_b)$ equals zero, which can be expressed as
	\begin{equation}\label{eq:2nd-der}
		\frac{\partial^2 \left(1-Q_1\left(a{z_b},br\right)\right)^{L_b}}{\partial {z_b^2}}=0.
	\end{equation}
	\textcolor{black}{It should be noted that an exact closed-form expression for this point is intractable. Hence, we propose a closed-form approximation to facilitate analysis.} Using another partial differential of Marcum $Q$-function in \cite[(5)]{app:b-2}, \eqref{eq:2nd-der} can be reformulated as
	\begin{multline}\label{eq:2nd-der-2}
		\left(1-L_b\right)be^{-\frac{\left(az_b\right)^2+(br)^2}{2}}\big[I_0\left(az_bbr\right)\big]^2-\left[1-Q_1\left(az_b,br\right)\right]\\
		\times\left\{-atI_1(abz_b)+\frac{b}{2}\left[I_0\left(az_bbr\right)+I_2\left(az_bbr\right)\right]\right\}=0.
	\end{multline}
	Denoting the solution to (\ref{eq:2nd-der-2}) as $\dot{z}_b$, for large $L_b$, the term $1-Q_1\left(\dot{z}_b,br\right)\in(0,1)$ can be approximated as 
	\begin{equation}
		1-Q_1\left(a\dot{z}_b,br\right)=\big[\,g_0(\dot{z}_b)\big]^{\frac{1}{L_b}}\approx 1.
	\end{equation}
	Using the above approximation and $I_n(x)\simeq \frac{e^x}{\sqrt{2\pi x}}$ for large $x$, \eqref{eq:2nd-der-2} can be equivalently written as
	\begin{equation}\label{eq:t-eqn}
		\left(2-L_b\right)e^{-\frac{\left(az_b-br\right)^2}{2}}+\left(1-\frac{az_b}{br}\right)\sqrt{2\pi az_bbr}=0.
	\end{equation}
	To address the transcendental equation, a first-order expansion around $z_b=\frac{br}{a}$ is applied on the second term of \eqref{eq:t-eqn}, which yields
	\begin{equation}\label{eq:t-eqn-2}
		e^{-\frac{\left(az_b-br\right)^2}{2}}-\frac{\sqrt{2\pi}}{2-L_b}(az_b-br)=0.
	\end{equation}
	The transition from \eqref{eq:t-eqn} to \eqref{eq:t-eqn-2} can be explained as follows. The solution to \eqref{eq:t-eqn} with respect to $z_b$ can be viewed as the intersection of the functions $g_3(z_b)=(L_b-2)e^{-\frac{\left(az_b-br\right)^2}{2}}$ and $g_4(z_b)=\left(1-\frac{az_b}{b}\right)\sqrt{2\pi az_bbr}$, denoted as $\dot{z}_b\in\left(0,\frac{br}{a}\right)$. As $L_b$ increases, the maximum value of $g_3(z_b)$ rises, enhancing its slope substantially. In this case, the value of $g_4(z_b)$ for $z_b\in\left(0,\dot{z}_b\right)$ becomes insignificant, thus requiring only an approximation of $g_4(z_b)$ near the center of $g_3(z_b)$, which is obtained by a first-order expansion around $\frac{br}{a}$.
	
	Note that \eqref{eq:t-eqn-2} can be effectively solved using the Lambert $W$-function, with the resulting expression provided by
	\begin{equation}\label{eq:rstar}
		\begin{split}
			\dot{z}_b&=\frac{1}{a}\left(br- {\sqrt{{\cal W}\left(\frac{\left(L_b-2\right)^2}{2\pi}\right)}}\right)\\
			&\color{black}{=\frac{1}{\mu^2}\left(r- {\underbrace{\sqrt{\frac{1-\mu^2}{2}{\cal W}\left(\frac{\left(L_b-2\right)^2}{2\pi}\right)}}_{\delta_{b}}}\right).}
		\end{split}
	\end{equation}
	Then replacing $\left(1-Q_1\left(a{z_b},br\right)\right)^{L_b}$ in \eqref{eq:F-lemma} with the Heaviside step function shifted at $\dot{z}_b$ in \eqref{eq:rstar} results in
	\begin{equation}\label{eq:F-hb}
		F_{\left|\hat{h}_b\right|}\left(r\right)=\!\int^{\infty}_{\dot{z}_b}2z_be^{-z_b^2}dz_b.
	\end{equation}
	Solving \eqref{eq:F-hb} gives \eqref{eq:F-constcorr}. \textcolor{black}{Recall that $z_b=x_{b,l_0}^2+y_{b,l_0}^2>0$. Therefore, a Heaviside step function $H\big(r-{\delta_b}\big)$ is applied to the lower bound of the integral, and the proof is complete.}

	\section{Proof of Proposition \ref{lemma:cdf-allblocks}}\label{app:c}
	Based on (\ref{eq:F-constcorr}), the CDF of $|\hat{h}|$ is given by
	\begin{subequations}
		\begin{align}
			F_{|\hat{h}|}&=\prod^{B}_{b=1}\left[1-e^{-\frac{1}{\mu^2}\left(r-\delta_{b}\right)^2}\right],\label{exact-cdf-blk}\\
			&\approx\left[1-e^{-\frac{1}{\mu^2}\left(r-\delta_{b}\right)^2}\right]^{B},\label{approx-cdf-blk}
		\end{align}
	\end{subequations}
	where the approximation has assumed $\delta_{1}=\cdots=\delta_{b}=\cdots=\delta_{B}$. Note that the upper and lower bounds of (\ref{exact-cdf-blk}) can be obtained by selecting ${\delta}_b=\max\limits_{b\in\mathcal{B}}{\delta_{b}}$ and ${\delta}_b=\min\limits_{b\in\mathcal{B}}{\delta_{b}}$, respectively.  For fairness, the arithmetic mean of (\ref{approx-cdf-blk}) is calculated over all $\delta_{b}$, $b\in{\mathcal{B}}$, given by (\ref{eq:cdf-allblocks}).
	
	\section{Proof of Lemma \ref{lemma:con-ser-approx}}\label{app:d}
	\textcolor{black}{Our proof starts from \eqref{eq:ser_ncoh}. Given $|\hat{h}|=r$, $X^{(\text{ncoh})}_{k|\,|\hat{h}|=r}$ is a Rician distributed random variable following $X^{(\text{ncoh})}_{k|\,|\hat{h}|=r}\sim \mathrm{Ri}\Big(\frac{|\xi_{m-k}|r}{M},\sqrt{\frac{N_0}{2}}\Big)$. Then the CDF of $\hat{X}^{(\text{ncoh})}_{|\,|\hat{h}|=r}$ is given by}
	\begin{equation}\label{exact-con-cdf-ncoh}\color{black}
		\begin{split}
			F_{{\hat{X}_{|\,|\hat{h}|=r}}^{(\rm{ncoh})}}(z)=\prod_{k=1\atop k\neq{m}}^{M}\left[1-Q_1\left(\frac{\left|\xi_{m-k}\right|{r}}{\sqrt{\frac{N_0}{2}}},\frac{{z}}{\sqrt{\frac{N_0}{2}}}\right)\right].
		\end{split}
	\end{equation}
	Note that \eqref{exact-con-cdf-ncoh} is the product of $M$ terms of special functions, which makes further analysis challenging. Our objective is to derive a more tractable approximation for this expression.
	
	We begin by analyzing the properties of the value $\big|\xi_{m-k}\big|$ in \eqref{eq:expsum}. An approximation is then presented in a more understandable form as 
	\begin{equation}
		\left|\xi_{m-k}\right|=\frac{1}{M}\frac{\sin\left(\frac{\pi \left(m-k\right)P}{M\cdot U}\right)}{\sin\left(\frac{\pi (m-k)}{M}\right)}\approx\frac{U}{M\cdot P}{\mathrm{sinc}\left(\frac{\pi \left(m-k\right)P}{M\cdot U}\right)}.
	\end{equation}
	According to the characteristics of $\mathrm{sinc}(x)$, $|\xi_{m-(m+1)}|=|\xi_{-1}|$ and $|\xi_{m-(m-1)}|=|\xi_{1}|$ are the two equal maximum values for a given $m$. When $\frac{P}{M\cdot U}$ increases, the difference between $|\xi_{m-(m{\pm{1}})}|$ and $|\xi_{m-k}|$, where $k\in\mathcal{K}$ and $k\neq m{\pm{1}}$, becomes more significant. Note that $M$ is typically large in LoRa systems. Based on this, we can separate $X^{(\rm{ncoh})}_{k|\,|\hat{h}|=r}$ into two parts as
	\begin{equation}
	\begin{split}
		\hat{X}^{(\rm{ncoh})}_{|\,|\hat{h}|=r}&={\max\left\{\max_{k\in\mathcal{K},\atop k\neq{m,m{\pm{1}}}}X^{(\rm{ncoh})}_{k|\,|\hat{h}|=r},\max_{k=m{\pm{1}}}X^{(\rm{ncoh})}_{k|\,|\hat{h}|=r}\right\}}\\
		&\triangleq\max\left\{\hat{X}^{(\rm{ncoh},\rm{N})}_{|\,|\hat{h}|=r},\hat{X}^{(\rm{ncoh},\rm{I})}_{|\,|\hat{h}|=r}\right\},
	\end{split}
	\end{equation}
	\textcolor{black}{where $\hat{X}^{(\rm{ncoh},\rm{N})}_{|\,|\hat{h}|=r}$ denotes that the value of these bins are dominated by the noise term, while $\hat{X}^{(\rm{ncoh},\rm{I})}_{|\,|\hat{h}|=r}$ are dominated by both two largest leakage interference and noise. Note that $F_{\hat{X}^{(\rm{ncoh})}_{|\,|\hat{h}|=r}}(z)=F_{\hat{X}^{(\rm{ncoh},\rm{N})}_{|\,|\hat{h}|=r}}(z)F_{\hat{X}^{(\rm{ncoh},\rm{I})}_{|\,|\hat{h}|=r}}(z)$, we have}
	
	\begin{multline}\label{cdf-n-i}\color{black}
		F_{\hat{X}^{(\rm{ncoh})}_{|\,|\hat{h}|=r}}(z)=\prod_{k=1,k\neq{m}\atop k\neq{{m}\pm 1}}^{M}\left[1-Q_1\left(\frac{\left|\xi_{m-k}\right|{r}}{\sqrt{\frac{N_0}{2}}},\frac{{z}}{\sqrt{\frac{N_0}{2}}}\right)\right]\\
		\times\prod_{k=m\pm 1}\left[1-Q_1\left(\frac{\left|\xi_{m-k}\right|{r}}{\sqrt{\frac{N_0}{2}}},\frac{{z}}{\sqrt{\frac{N_0}{2}}}\right)\right]
	\end{multline}
	Here, we first consider $\hat{X}^{(\rm{ncoh},\rm{N})}_{|\,|\hat{h}|=r}$. Note that the leakage interference is relatively small at these points, and we can directly neglect them. Therefore, setting $\big|\xi_{m-k}\big|=0$ for ${X}^{(\rm{ncoh},\rm{N})}_{k||\hat{h}|=r}$, $k\in\mathcal{K}$, $k\neq\{m,m\pm1\}$ and utilizing $Q_1(0,x)=e^{-\frac{x^2}{2}}$, we have
	\begin{equation}\label{eq:oneminusQ}
		\begin{split}
			F_{\hat{X}^{(\rm{ncoh},\rm{N})}_{k||\hat{h}|=r}}(z)&\approx\left(1-e^{-\frac{z^2}{N_0}}\right)^{M-3}\color{black}{\equiv F_{\hat{X}^{(\rm{ncoh},\rm{N})}_{k}}(z)}.
		\end{split}
	\end{equation}
	\textcolor{black}{Similar to Appendix \ref{app:b}, for large $M$, $F_{\hat{X}^{(\rm{ncoh},\rm{N})}_{k||\hat{h}|=r}}(z)$ also can be approximate as a Heaviside step function. Here, we proposed another approach to calculate the step point $\dot{z}$. Specifically, the step point can be obtained by calculating the area of the geometric shape enclosed by the complementary CDF of $F_{\hat{X}^{(\rm{ncoh},\rm{N})}_{|\,|\hat{h}|=r}}(z)$ and the coordinate axes. This is because for $z<\dot{z}$, $1-F_{\hat{X}^{(\rm{ncoh},\rm{N})}_{|\,|\hat{h}|=r}}(z)=1$ is the height of the area and the area can be equated to the length of the shape. Then we have} 
	\begin{equation}\label{eq:sigma-M3-ncoh}
		\begin{split}
			\dot{z}\triangleq&\,\sigma^{(\rm{ncoh},\rm{N})}_{M-3}=\int^{\infty}_0\left[1-F_{\hat{X}^{(\rm{ncoh})}_{|\,|\hat{h}|=r}}(z)\right]dz\\
			&=\int^{\infty}_0\left[1-\left(1-e^{-\frac{z^2}{N_0}}\right)^{M-3}\right]dz\\
			&\overset{(a)}{\approx}\sqrt{N_0\ln(M-3)}\left(1+\frac{\gamma}{2\ln(M-3)}\right),
		\end{split}
	\end{equation}
	where $(a)$ is given in \cite[(15) \& (22)]{app:d-1}. $\gamma\approx0.5772$ is the Euler-Mascheroni constant. Then we can rewrite \eqref{eq:oneminusQ} as
	\begin{equation}\label{step-func}
		\begin{split}
			F_{\hat{X}^{(\rm{ncoh},\rm{N})}_{|\,|\hat{h}|=r}}(z)\approx H\left(z-\sigma^{(N_0)}_{M-3}\right)\equiv F_{\hat{X}^{(\rm{ncoh},\rm{N})}_{}}(z),
		\end{split}
	\end{equation}
	Now we consider the term $\hat{X}^{(\rm{ncoh},\rm{I})}_{|\,|\hat{h}|=r}$. Utilizing the approximation of $\left(1-Q_n(\cdot)\right)^N\approx1-NQ_1(\cdot)$, we have the approximated CDF of $\hat{X}^{(\rm{ncoh},\rm{I})}_{|\,|\hat{h}|=r}$ as
	\begin{equation}\label{eq:Fapprox}
		\begin{split}
			F_{X^{(\rm{ncoh},\rm{I})}_{|\,|\hat{h}|=r}}(z)&=\left[1\!-Q_1\!\left(\frac{\big|\xi_{{\pm{1}}}\big|}{\sqrt{\frac{N_0}{2}}}{r},\frac{{z}}{\sqrt{\frac{N_0}{2}}}\right)\right]^2\\
			&\approx1\!-2Q_1\!\left(\frac{\big|\xi_{{\pm{1}}}\big|}{\sqrt{\frac{N_0}{2}}}{r},\frac{{z}}{\sqrt{\frac{N_0}{2}}}\right).
		\end{split}
	\end{equation}
	According to \eqref{cdf-n-i}, \eqref{eq:oneminusQ} and \eqref{eq:Fapprox}, (\ref{con-cdf-x-approx}) is obtained.
	\section{Proof of Lemma \ref{lemma:marcum-q-approx}}\label{app:e}
	According to \cite{app:e-1}, the first term asymptotic approximation of the Marcum $Q$-function $Q_n(x,y)$ for $y>x$ is given by
	\begin{equation}\label{eq:Q-1st}
		Q_n(x,y)\approx\left(\frac{y}{x}\right)^{n-\frac{1}{2}}Q(y-x),
	\end{equation}
	
	To avoid including non-integer power, the arithmetic mean of the bounds of generalized Marcum $Q$-function is used before (\ref{eq:Q-1st}), given by \cite{app:e-2}
	\begin{equation}\label{eq:Qapprox}
		Q_n(x,y)\approx\frac{1}{2}\left(Q_{n-0.5}(x,y)+Q_{n+0.5}(x,y)\right).
	\end{equation}
	Based on (\ref{eq:Q-1st}), (\ref{eq:Qapprox}) and the improved Chernoff bound for Gaussian $Q$-function, $Q(x)\approx \frac{1}{2}e^{-\frac{x^2}{2}}$, (\ref{eq:Q1-approx}) is obtained. 
	
	\section{Proof of Lemma \ref{lemma:con-ser-approx-cl}}\label{app:f}
	\textcolor{black}{Before the proof, we introduce the following integral.}
	\begin{equation}\label{int-pow-exp}\color{black}
		\begin{split} 
			J_1&=\int_{a}^{\infty}xe^{-(bx^2+cx)}dx\\  
			&=\frac{1}{2b}\int_{a}^{\infty}(2bx+c)e^{-(bx^2+cx)}dx-\frac{c}{2b}\int_{a}^{\infty}e^{-(bx^2+cx)}dx\\ 
			&=\frac{1}{2b}\int_{a^2b+ac}^{\infty}e^{-t}dt-\frac{c}{2b}\underbrace{\int_{a}^{\infty}e^{-(bx^2+cx)}dx}_{J_0}\\ 
			&=\frac{1}{2b}e^{-a(ab+c)}-\frac{c}{2b}\frac{e^{-\frac{c^2}{2b}}}{\sqrt{b}}\int_{a\sqrt{b}+\frac{c}{\sqrt{2b}}}^{\infty}e^{-t^2}dt\\ 
			&=\frac{1}{2b}e^{-a(ab+c)}-\underbrace{\frac{c}{2b}\frac{e^{-\frac{c^2}{2b}}}{\sqrt{b}}\mathrm{erfc}\left(\frac{ab\sqrt{2}+c}{\sqrt{2b}}\right)}_{J_0},
		\end{split}
	\end{equation}
	
	\begin{figure*}[!ht]
		\begin{align}\label{eq:X}
			&\bar{A}_1^{(\rm{ncoh})}=\int^{\infty}_{\hat{\delta}}\left(r-{\delta}_b\right)e^{-\frac{\tilde{b}}{\mu^2}\left(r-{\delta}_b\right)^2}\mathrm{erfc}\left(\frac{ {{\sigma}_{M-3}^{(\rm{ncoh})}\!-D r}}{\sqrt{N_0}} \right)dr\\ 
			&=\frac{e^{-\frac{\tilde{b}}{\mu^2}\left(\hat{\delta}-{\delta}_b\right)^2}}{2\frac{\tilde{b}}{\mu^2}}\,\mathrm{erfc}\left(\frac{ {{\sigma}_{M-3}^{(\rm{ncoh})}\!-D \hat{\delta}}}{\sqrt{N_0}}\right)
			+\frac{{D}}{2\frac{\tilde{b}}{\mu^2}\sqrt{{D^2}+{\frac{\tilde{b}}{\mu^2}N_0}}}
			e^{-\frac{\frac{\tilde{b}}{\mu^2}\left({\sigma_{M-3}^{(\rm{ncoh})}}-D\delta_b\right)^2}{D^2+\frac{\tilde{b}}{\mu^2}N_0}}\mathrm{erfc}\left( \frac{\frac{D^2}{N_0}\hat{\delta}-\frac{D}{N_0}{\sigma_{M-3}^{(\rm{ncoh})}}+\frac{\tilde{b}}{\mu^2}\left(\hat{\delta}-\delta_b\right)}{\sqrt{\frac{\tilde{b}}{\mu^2}+\frac{D^2}{N_0}}}\right) \notag
		\end{align}
		\hrule
	\end{figure*}
	
	\begin{figure*}[!ht]
		\vspace{-0.5cm}
		\begin{multline}\label{eq:Y}
			\bar{A}_2^{(\rm{ncoh})}\approx\int^{\infty}_{\hat{\delta}}\!\!\!
			e^{-\frac{\tilde{b}}{\mu^2}\left(r-\delta_b\right)^2}e^{-\frac{|\xi_{\pm1}|^2+D^2}{N_0}{r^2}+\frac{2\sigma_{M-3}^{(\rm{ncoh})}\left(|\xi_{\pm1}|+D\right)}{N_0}r-\frac{{2}\left(\sigma_{M-3}^{(\rm{ncoh})}\right)^2}{N_0}} dr\\
			=\frac{\sqrt{\frac{\pi}{2}}}{\sqrt{\frac{\tilde{b}}{\mu^2}+{\frac{D^2+|\xi_{\pm1}|^2}{N_0}}  }}e^{ -\frac{2{\left({\sigma}^{(\rm{ncoh})}_{M-3}\right)^2}}{N_0}-\frac{\tilde{b}\delta^2_b}{\mu^2}+
				\frac{ \left( \frac{\tilde{b}\delta_b}{\mu^2}+{\frac{\sigma^{(\rm{ncoh})}_{M-3}\left(D+|\xi_{\pm1}|\right)}{{N_0}}} \right)^2 }{\sqrt{\frac{\tilde{b}}{\mu^2}+{\frac{D^2+|\xi_{\pm1}|^2}{N_0}}}}}
			\!{\rm{erfc}}\!\left(-\frac{ \frac{\tilde{b}\delta_b}{\mu^2}+{\frac{\sigma^{(\rm{ncoh})}_{M-3}\left(D+|\xi_{\pm1}|\right)}{{N_0}}}
				-\left(\frac{\tilde{b}}{\mu^2}+{\frac{D^2+|\xi_{\pm1}|^2}{N_0}}\right)\hat{\delta}}{\sqrt{\frac{\tilde{b}}{\mu^2}+{\frac{D^2+|\xi_{\pm1}|^2}{N_0}}  } }\right)
		\end{multline}
		\hrule
	\end{figure*}
	
	\begin{figure*}[!ht]
		\vspace{-0.5cm}
		\begin{equation}\label{eq:Z}
			\begin{split}
				\bar{A}_3^{(\rm{ncoh})}&=\int^{\infty}_{\hat{\delta}}\left(r-\delta_b\right)e^{-\frac{\tilde{b}}{\mu}\left(r-\delta_b\right)^2-{\frac{\left({D}-|\xi_{\pm1}|\right)^2}{2N_0}r^2}}\mathrm{erfc}\left(\frac{2{\sigma}^{(\rm{ncoh})}_{M-3}-\left(D+|\xi_{\pm1}|\right)r}{\sqrt{2N_0}}\right)dr\\
				&\approx {\Xi}_1^{-\frac{3}{2}}e^{\frac{{\Xi}^2_{2}}{4{\Xi}_{1}}-{\Xi}_{3}}\left(\sqrt{{\Xi}_{1}}e^{-\frac{\left(\frac{{\Xi}_{2}}{2}+{\Xi}_{1}\hat{\delta}\right)}{2{\Xi}_{1}}}-{\sqrt{\pi}}\left(\frac{{\Xi}_{2}}{2}+{\Xi}_{1}\delta_b\right)\mathrm{erfc}\left(\frac{\frac{{\Xi}_{2}}{2}+{\Xi}_{1}\hat{\delta}}{\sqrt{{\Xi}_{1}}}\right)\right)
			\end{split}
		\end{equation}
		\hrule
	\end{figure*}
	\noindent\textcolor{black}{where $J_0$ is obtained by using the definition of $\mathrm{erfc}(\cdot)$. Note that \eqref{eq:con-ser-approx} has the same structure as $J_0$ and $J_1$, which includes a power term, an exponential term and a lower limit on the integral. Then \eqref{eq:con-ser-approx-cl} can be derived directly through variable substitution.}
	
	\section{Proof of Proposition \ref{prop:ser-ncoh}}\label{app:g}
	\textcolor{black}{To obtain the average SER, we integral the conditional SER \eqref{eq:con-ser-approx-cl} on the PDF of Rx-SISO-FAS channel \eqref{eq:f-hstar}. After simplification, we obtain \eqref{ser-ncoh}, where $\bar{A}_1^{\rm{ncoh}}$, $\bar{A}_2^{\rm{ncoh}}$ and $\bar{A}_3^{\rm{ncoh}}$ are three integral give as \eqref{eq:X}, \eqref{eq:Y} and \eqref{eq:Z}, respectively. Moreover, $\Xi_1$, $\Xi_2$ and $\Xi_3$ in $\bar{A}_3^{\rm{ncoh}}$ are defined as}
	\begin{equation}
		\left\{\begin{aligned}
			&{\Xi}_1=\frac{\tilde{b}}{\mu^2}+\frac{\left(D-|\xi_{\pm1}|\right)^2}{2N_0}+\frac{\Theta_1\left(D+{|\xi_{\pm1}|}\right)^2}{2N_0}\\
			&{\Xi}_2=-2\frac{\tilde{b}\delta_b}{\mu^2}-{2\Theta_1\frac{{\sigma}^{(\rm{ncoh})}_{M-3}}{N_0}\left(D+|\xi_{\pm1}|\right)}-\frac{\Theta_2\left(D+|\xi_{\pm1}|\right)}{\sqrt{2N_0}}\\
			&{\Xi}_3=\frac{\tilde{b}\delta_b^2}{\mu^2}+{2\Theta_1\frac{\left({\sigma}^{(\rm{ncoh})}_{M-3}\right)^2}{N_0}}+{\sqrt{\frac{2}{N_0}}\Theta_2{\sigma}^{(\rm{ncoh})}_{M-3} }+\Theta_3.
		\end{aligned}\right.
	\end{equation}	
	\textcolor{black}{where $\Theta_1=0.7640$, $\Theta_2=0.7640\sqrt{2}$ and $\Theta_3=0.6964$. Note that $\bar{A}_1^{(\rm{ncoh})}$ is obtained using integral by part. The approximation in $\bar{A}_2^{(\rm{ncoh})}$ is due to $\frac{r-\sigma_b}{r}\approx1$ as the value of $\sigma_b$ is small. Then integral of $\bar{A}_2^{(\rm{ncoh})}$ is derived using $J_0$ in \eqref{int-pow-exp}. Finally, the approximation in $\bar{A}_3^{(\rm{ncoh})}$ is based on substituting the $\mathrm{erfc}(\cdot)$ by $\mathrm{erfc(x)}\approx2{e}^{\Theta_1x^2+\Theta_2x+\Theta_3}$  \cite[(8) \& Table I]{app:f-1}, and the integral of $\bar{A}_3^{(\rm{ncoh})}$ is derived using $J_1$ in \eqref{int-pow-exp}.}

	\vfill
	
\end{document}